\documentclass[journal]{IEEEtran}
\usepackage{amsmath}
\usepackage{amssymb}
\usepackage{upref}
\usepackage{amsfonts}
\usepackage{graphicx}
\usepackage{verbatim}
\usepackage{pict2e}

\newcommand{\field}[1]{\mathbb{#1}}

\newcommand{\Z}{\field{Z}}

\newcommand{\R}{\field{R}}

\newcommand{\cA}{{\cal A}}
\newcommand{\cB}{{\cal B}}
\newcommand{\cC}{{\cal C}}

\newcommand{\cG}{{\cal G}}

\newcommand{\cP}{{\cal P}}
\newcommand{\cQ}{{\cal Q}}
\newcommand{\cS}{{\cal S}}

\newcommand{\cU}{{\cal U}}
\newcommand{\cT}{{\cal T}}

\newcommand{\sP}{\cP}
\newcommand{\sG}{\cG}

\newcommand{\Gr}{\smash{{\sG\kern-1.5pt}_q\kern-0.5pt(n,k)}}
\newcommand{\Grtwo}{\smash{{\sG\kern-1.5pt}_2\kern-0.5pt(n,k)}}
\newcommand{\Gkone}{\smash{{\sG\kern-1.5pt}_q\kern-0.5pt(n,k_1)}}
\newcommand{\Gktwo}{\smash{{\sG\kern-1.5pt}_q\kern-0.5pt(n,k_2)}}
\newcommand{\Ps}{\smash{{\sP\kern-2.0pt}_q\kern-0.5pt(n)}}

\newcommand{\bi}{{\bf i}}

\newcommand{\deff}{\mbox{$\stackrel{\rm def}{=}$}}

\newtheorem{theorem}{Theorem}

\newtheorem{lemma}{Lemma}
\newtheorem{remark}{Remark}
\newtheorem{cor}{Corollary}
\newtheorem{example}{Example}

\begin{document}

\bibliographystyle{IEEEtran}

\title{Sequence Folding, Lattice Tiling, \\ and Multidimensional Coding}
\author{Tuvi Etzion,~\IEEEmembership{Fellow,~IEEE}
\thanks{T. Etzion is with the Department of Computer Science,
Technion --- Israel Institute of Technology, Haifa 32000, Israel.
(email: etzion@cs.technion.ac.il).}
\thanks{The material in this paper was presented in part in the 2009 IEEE
Information Theory Workshop, Taormina, Sicily, Italy, October
2009.}
\thanks{This work was supported in part by the United States-Israel
Binational Science Foundation (BSF), Jerusalem, Israel, under
Grant 2006097.} }

\maketitle
\begin{abstract}
Folding a sequence $S$ into a multidimensional box is a well-known
method which is used as a multidimensional coding technique. The
operation of folding is generalized in a way that the sequence $S$
can be folded into various shapes and not just a box. The new
definition of folding is based on a lattice tiling for the given
shape $\cS$ and a direction in the $D$-dimensional integer grid.
Necessary and sufficient conditions that a lattice tiling for
$\cS$ combined with a direction define a folding of a sequence
into $\cS$ are derived. The immediate and most impressive
application is some new lower bounds on the number of dots in
two-dimensional synchronization patterns. This can be also
generalized for multidimensional synchronization patterns. The
technique and its application for two-dimensional synchronization
patterns, raise some interesting problems in discrete geometry. We
will also discuss these problems. It is also shown how folding can
be used to construct multidimensional error-correcting codes.
Finally, by using the new definition of folding, multidimensional
pseudo-random arrays with various shapes are generated.
\end{abstract}


\section{Introduction}
\label{sec:introduction}

Multidimensional coding in general and two-dimensional coding in
particular are subjects which attract lot of attention in the last
three decades. One of the main reasons is their modern
applications which have developed during these years. Such
applications for synchronization patterns include radar, sonar,
physical alignment, and time-position synchronization. For
error-correcting codes they include two-dimensional magnetic and
optical recording as well as three-dimensional holographic
recording. These are the storage devices of the future.
Applications for pseudo-random arrays include scrambling of
two-dimensional data, two-dimensional digital watermarking, and
structured light patterns for imaging systems. Each one of these
structures (multidimensional synchronization patterns,
error-correcting array codes, and pseudo-random arrays), and its
related coding problem, is a generalization of an one-dimensional
structure. But, although the related theory of the one-dimensional
case is well developed, the theory for the multidimensional case
is developed rather slowly. This is due that the fact the most of
the one-dimensional techniques are not generalized easily to
higher dimensions. Hence, specific techniques have to be developed
for multidimensional coding. One approach in multidimensional
coding is to take an one-dimensional code and to transform it into
a multidimensional code. One technique in this approach is called
folding and it is the subject of the current paper. This technique
was applied previously for two-dimensional synchronization
patterns, for pseudo-random arrays, and lately for
multidimensional error-correcting codes. We start with a short
introduction to these three multidimensional coding problems which
motivated our interest in the generalization of folding.

\noindent {\bf Synchronization patterns}

One-dimensional synchronization patterns were first introduced by
Babcock in connection with radio interference~\cite{Bab}. Other
applications are discussed in details in~\cite{BlGo77} and some
more are given in~\cite{ASU,LaSa88}. The two-dimensional
applications and related structures were first introduced
in~\cite{GoTa82} and discussed in many papers,
e.g.~\cite{GoTa84,Rob85,Games87,BlTi88,Rob97}. The two-dimensional
problems has also interest from discrete geometry point of view
and it was discussed for example in~\cite{EGRT92,LeTh95}. Recent
new application in keys predistribution for wireless sensor
networks~\cite{BEMP} led to new related two-dimensional problems
concerning these patterns which are discussed
in~\cite{BEMP08a,BEMP08b}. It has raised the following discrete
geometry problem: given a regular polygon with area $s$ on the
square (or hexagonal) grid, what is the maximum number of grid
points that can be taken, such that any two lines connecting these
grid points are different either in their length or in their
slope. Upper bound technique based on an idea of Erd\"{o}s and
Tur\'{a}n~\cite{EGRT92,Erd41} is given in~\cite{BEMP08a}. Some
preliminary lower bounds on the number of dots are also given
in~\cite{BEMP08a}, where the use of folding is applied. Folding
for such patterns was first used by~\cite{Rob97}. An
one-dimensional ruler was presented as a binary sequence and
written into a two-dimensional array row by row, one binary symbol
to each entry of the array. This was generalized for higher
dimensions, say $n_1 \times n_2 \times n_3$ array, by first
partitioning the array into $n_1$ two-dimensional arrays of size
$n_2 \times n_3$. The one-dimensional sequence is written into the
these $n_2 \times n_3$ arrays one by one in the order defined by
the three-dimensional array. To each of these $n_2 \times n_3$
arrays the sequence is written row by row. Folding into higher
dimensions is done similarly and can be defined recursively. This
technique was used in~\cite{Rob97} to generate asymptotically
optimal high dimensional synchronization patterns.

\noindent {\bf Error-correcting codes}

There is no need for introduction to one-dimensional
error-correcting codes. Two-dimensional and multidimensional
error-correcting codes were discussed by many authors,
e.g.~\cite{Ima73,Abd86,AMT,BBZS,BBV,EtVa02,ScEt05,Boy06,EtYa09}.
Multidimensional error-correcting codes are of interest when the
errors are not random errors. For correction of up to $t$ random
errors in a multidimensional array, we can consider the elements
in the array as an one-dimensional sequence and use a
$t$-error-correcting code to correct these errors. Hence, when we
talk about multidimensional error-correcting codes we refer to the
errors as special ones such as the rank of the error
array~\cite{Gab85,Rot91}, or crisscross patterns~\cite{Rot91,
Rot97,BlaBru00}, etc. An important family of multidimensional
error-correcting codes are the burst-error-correcting codes. In
these codes, we assume that the errors are contained in a cluster
whose size is at most $b$. The one-dimensional case was considered
for more than forty years. Fire~\cite{Fire59} was the first to
present a general construction. Optimal burst-correcting codes
were considered in~\cite{Abr59,ElSh,AMOT}. Generalizations,
especially for two-dimensional codes, but also for
multidimensional codes were considered in various research papers,
e.g.~\cite{Abd86,AMT,BBV,EtVa02,ScEt05,EtYa09}. In general,
"simple" folding of one-dimensional codes were not considered for
multidimensional error-correcting codes. Even so in many of these
papers, one-dimensional burst-correcting codes and
error-correcting codes, were transferred into two dimensional
codes, e.g.~\cite{BBZS,BBV,EtVa02,ScEt05,Boy06,EtYa09}. Colorings
for two-dimensional coding, which transfer one-dimensional codes
into multidimensional arrays were considered for interleaving
schemes~\cite{BBV} and other techniques~\cite{EtYa09}. These
colorings can be compared to the coloring which will be used in
the sequel for folding. There is another related problem of
generating an array in which burst-errors can be corrected on an
unfolded sequence generated from the
array~\cite{FaHo,BFvT,Bla90,ZhWo,Zha91}.

\noindent {\bf Pseudo-random arrays}

The one-dimensional pseudo-random sequences are the maximal length
linear shift register sequences known as M-sequences and also
pseudo-noise (PN) sequences~\cite{Golomb}. These are sequences of
length $2^n-1$ generated by a linear feedback shift-register of
order $n$. They have many desired properties such as

\begin{itemize}
\item Recurrences Property - the entries satisfy a recurrence
relation of order $n$.

\item Balanced Property - $2^{n-1}$ entries in the sequence are
{\it ones} and $2^{n-1}-1$ entries in the sequence are {\it
zeroes}.

\item shift-and-Add Property - when a sequence is added bitwise to
its cyclic shift another cyclic shift of the sequence is obtained.

\item Autocorrelation Property - the out-of-phase value of the
autocorrelation function is always -1.

\item Window Property - each nonzero $n$-tuple appears exactly
once in one period of the sequence.
\end{itemize}
There are other properties which we will not
mention~\cite{McSl76}. For a comprehensive work on these sequences
the reader is referred to~\cite{Golomb}. Related sequences are the
de Bruijn sequences of length $2^n$ which are generated by
nonlinear feedback shift-register of order $n$. These sequences
have the window property, i.e.,  each $n$-tuple appears exactly
once in one period of the sequence.

The two-dimensional generalizations of pseudo-noise and de Bruijn
sequences are the pseudo-random arrays and perfect
maps~\cite{McSl76,NMIF,ReSt,FFMS,Etz88,Pat94}. Pseudo-random
arrays were also called {\it linear recurring arrays having
maximum=area matrices} by Nomura, Miyakawa, Imai, and
Fukuda~\cite{NMIF} who were the first to construct them. Perfect
maps and pseudo-random arrays have been used in two-dimensional
range-finding, in data scrambling, and in various kinds of mask
configurations. More recently, pseudo-random arrays have found
other applications in new and emerging technological areas. One
such application is robust, undetectable, digital watermarking of
two-dimensional test images~\cite{STO,TSO}. Another interesting
example is the use of pseudo-random arrays in creating
\emph{structured light}, which is a new reliable technique for
recovering the surface of an object. The structured-light
technique is based on projecting a light pattern and observing the
illuminated scene from one or more points of
view~\cite{Hsi01,MOCDZN,SPB,PSCF}. As mentioned in these papers,
this technique can be generalized to three dimensions; hence,
constructions of three-dimensional perfect maps and pseudo-random
arrays are also of interest.

\vspace{0.2cm}

The main goal of this paper is to generalize the well-known
technique, folding, for generating multidimensional codes of these
types, synchronization patterns, burst-correcting codes, and
pseudo-random arrays. The generalization will enable to obtain the
following results:

\begin{enumerate}
\item Form new two-dimensional codes for these applications.

\item Generalize all the multidimensional codes for any number of
dimensions in a simple way.

\item Form some optimal codes not known before.

\item Make these codes feasible not just for multidimensional
boxes, but also for many other different shapes.

\item Solve the synchronization pattern problem as a discrete
geometry problem for various two-dimensional shapes, and in
particular regular polygons.
\end{enumerate}

It is important to note that folding which was used in other
places in the literature aim only at one goal. Our folding aim is
at several goals. Even so, our description of folding is simple
and very intuitive for all these goals.

The rest of this paper is organized as follows. In
Section~\ref{sec:FoldTile} we define the basic concepts of folding
and lattice tiling. Tiling and lattices are basic combinatorial
and algebraic structures. We will consider only integer lattice
tiling. We will summarize the important properties of lattices and
lattice tiling. In Section~\ref{sec:folding} we will present the
generalization of folding into multidimensional shapes. All
previous known folding definitions are special cases of the new
definition. The new definition involves a lattice tiling and a
direction. We will prove necessary and sufficient conditions that
a lattice with a direction define a folding. We first present a
proof for the two-dimensional case since it is the most applicable
case. We continue to show the generalization for the
multidimensional case. For the two-dimensional case the proofs are
slightly simpler than the slightly different proofs for the
multidimensional case. we will first consider folding in which two
consecutive elements in the folded sequence are also adjacent, at
least cyclically, in the array. This will be generalized to
folding in which each two consecutive elements in the folded
sequence are not necessarily adjacent in the array. In
Section~\ref{sec:bounds} we give a short summary on
synchronization patterns and present basic theorems concerning the
bounds on the number of elements in such patterns. In
Section~\ref{sec:DDAs} we apply the results of the previous
sections to obtain new type of synchronization patterns which are
asymptotically either optimal or almost optimal. In
Section~\ref{sec:hex} we discuss folding in the hexagonal grid and
present construction for synchronization patterns in this grid
with shapes of hexagons or circles. In Section~\ref{sec:ECC} we
show how folding can be applied to construct multidimensional
error-correcting codes. In section~\ref{sec:pseudo-random} we
generalize the constructions in~\cite{NMIF,McSl76} to form
pseudo-random arrays on different multidimensional shapes.
Conclusion and problems for further research are given in
Section~\ref{sec:conclude}.

\section{Folding and Lattice Tiling}
\label{sec:FoldTile}

\subsection{Folding}

Folding a rope, a ruler, or any other feasible object is a common
action in every day life. Folding an one-dimensional sequence into
a $D$-dimensional array is very similar, but there are a few
variants. First, we will summarize three variants for folding of
an one-dimensional sequence $s_0 s_1 \cdots s_{m-1}$ into a
two-dimensional array $\cA$. The generalization for a
$D$-dimensional array is straightforward while the description
becomes more clumsy.

\begin{enumerate}
\item[\bf F1.] $\cA$ is considered as a cyclic array horizontally
and vertically in such a way that a walk diagonally visits all the
entries of the array. The elements of the sequence are written
along the diagonal of the $r \times t$ array $\cA$. This folding
works (i.e., all elements of the sequence are written into the
array) if and only if $r$ and $t$ are relatively primes.

\item[\bf F2.] The elements of the sequence are written row by row
(or column by column) in $\cA$.

\item[\bf F3.] The elements of the sequence are written diagonal
by diagonal in $\cA$.
\end{enumerate}

\begin{example}
$~$

\vspace{0.2cm}

\noindent Example for {\bf F1:}

\vspace{0.2cm}

Given the M-sequence $000111101011001$ of length 15, we fold it
into a $3 \times 5$ array with a $2 \times 2$ window property (the
extra row and extra column are given for better understanding of
the folding).

\vspace{0.2cm}

\begin{center}
$\begin{array}{|c|c|c|c|c||c|}
\hline 0&6&12&3&9&0 \\
\hline \hline 5&11&2&8&14&5 \\
\hline
10&1&7&13&4&10 \\
\hline 0&6&12&3&9&0 \\ \hline
\end{array}$
\end{center}

\vspace{0.5cm}

\begin{center}
$\begin{array}{|c|c|c|c|c||c|}
\hline 0&1&0&1&0&0 \\
\hline \hline 1&1&0&1&1&1 \\
\hline
1&0&0&0&1&1 \\
\hline 0&1&0&1&0&0 \\ \hline
\end{array}$
\end{center}

\vspace{0.2cm}

\noindent Example for {\bf F2:}

\vspace{0.2cm}

The following sequence (ruler) of length 13 with five dots is
folded into a $3 \times 5$ array

\vspace{0.2cm}

\begin{scriptsize}
\begin{center}
$\begin{array}{|c|c|c|c|c|c|c|c|c|c|c|c|c|} \hline 0&1&2
&3&4&5&6&7&8&9&10&11&12 \\
\hline \bullet&\bullet& & &\bullet& & & & & & \bullet& &\bullet  \\
\hline
\end{array}$
\end{center}
\end{scriptsize}

\vspace{0.5cm}

\begin{center}
$\begin{array}{|c|c|c|c|c|}
\hline 10&11&12&13&14 \\
\hline
5&6&7&8&9 \\
\hline 0&1&2&3&4 \\ \hline
\end{array}$
\end{center}

\vspace{0.5cm}

\begin{center}
$\begin{array}{|c|c|c|c|c|}
\hline \bullet&&\bullet&$~$& \\
\hline
&&&& \\
\hline \bullet&\bullet&&&\bullet \\ \hline
\end{array}$
\end{center}

\vspace{0.5cm}

\noindent Example for {\bf F3:}

\vspace{0.2cm}

The following $B_2$-sequence in $\Z_{31}$ : $\{ 0,1,4,10,12,17 \}$
(can be viewed as a cyclic ruler) is folded into an infinite array
(we demonstrate part of the array with folding into a small
rectangle is given in bold). Note, that while the folding is done
we should consider all the integers modulo 31 (see
Figure~\ref{fig:fold_diag}).

\vspace{0.2cm}

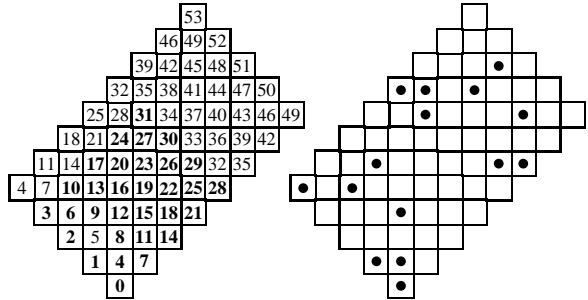
\begin{figure}[t]

\centering \setlength{\unitlength}{.65mm}
\begin{scriptsize}
\begin{picture}(35,50)(40,20)
\put(5,20){\framebox(35,5){}} \put(10,15){\framebox(25,15){}}
\put(15,10){\framebox(15,25){}} \put(20,5){\framebox(5,35){}}

\put(25,40){\framebox(35,5){}} \put(30,35){\framebox(25,15){}}
\put(35,30){\framebox(15,25){}} \put(40,25){\framebox(5,35){}}

\put(0,25){\framebox(5,5){}} \put(5,30){\framebox(5,5){}}
\put(10,35){\framebox(5,5){}} \put(15,40){\framebox(5,5){}}
\put(20,45){\framebox(5,5){}} \put(25,50){\framebox(5,5){}}
\put(30,55){\framebox(5,5){}} \put(35,60){\framebox(5,5){}}

\put(20,35){\makebox(5,5){\bf 24}}

\put(15,30){\makebox(5,5){\bf 17}} \put(20,30){\makebox(5,5){\bf
20}} \put(25,30){\makebox(5,5){\bf 23}}

\put(10,25){\makebox(5,5){\bf 10}} \put(15,25){\makebox(5,5){\bf
13}} \put(20,25){\makebox(5,5){\bf 16}}
\put(25,25){\makebox(5,5){\bf 19}} \put(30,25){\makebox(5,5){\bf
22}}

\put(0,25){\makebox(5,5){4}} \put(5,30){\makebox(5,5){11}}
\put(10,35){\makebox(5,5){18}} \put(15,40){\makebox(5,5){25}}
\put(20,45){\makebox(5,5){32}} \put(25,50){\makebox(5,5){39}}
\put(30,55){\makebox(5,5){46}} \put(35,60){\makebox(5,5){53}}
\put(5,25){\makebox(5,5){7}} \put(10,30){\makebox(5,5){14}}
\put(15,35){\makebox(5,5){21}} \put(20,40){\makebox(5,5){28}}
\put(25,45){\makebox(5,5){35}} \put(30,50){\makebox(5,5){42}}
\put(35,55){\makebox(5,5){49}}

\put(5,20){\makebox(5,5){\bf 3}} \put(10,20){\makebox(5,5){\bf 6}}
\put(15,20){\makebox(5,5){\bf 9}} \put(20,20){\makebox(5,5){\bf
12}} \put(25,20){\makebox(5,5){\bf 15}}
\put(30,20){\makebox(5,5){\bf 18}} \put(35,20){\makebox(5,5){\bf
21}}

\put(10,15){\makebox(5,5){\bf 2}} \put(15,15){\makebox(5,5){5}}
\put(20,15){\makebox(5,5){\bf 8}} \put(25,15){\makebox(5,5){\bf
11}} \put(30,15){\makebox(5,5){\bf 14}}

\put(15,10){\makebox(5,5){\bf 1}} \put(20,10){\makebox(5,5){\bf
4}} \put(25,10){\makebox(5,5){\bf 7}}

\put(20,5){\makebox(5,5){\bf 0}}

\put(40,55){\makebox(5,5){52}}

\put(35,50){\makebox(5,5){45}} \put(40,50){\makebox(5,5){48}}
\put(45,50){\makebox(5,5){51}}

\put(30,45){\makebox(5,5){38}} \put(35,45){\makebox(5,5){41}}
\put(40,45){\makebox(5,5){44}} \put(45,45){\makebox(5,5){47}}
\put(50,45){\makebox(5,5){50}}

\put(25,40){\makebox(5,5){\bf 31}} \put(30,40){\makebox(5,5){34}}
\put(35,40){\makebox(5,5){37}} \put(40,40){\makebox(5,5){40}}
\put(45,40){\makebox(5,5){43}} \put(50,40){\makebox(5,5){46}}
\put(55,40){\makebox(5,5){49}}

\put(30,35){\makebox(5,5){\bf 30}} \put(35,35){\makebox(5,5){33}}
\put(40,35){\makebox(5,5){36}} \put(45,35){\makebox(5,5){39}}
\put(50,35){\makebox(5,5){42}}

\put(35,30){\makebox(5,5){\bf 29}} \put(40,30){\makebox(5,5){32}}
\put(45,30){\makebox(5,5){35}}

\put(40,25){\makebox(5,5){\bf 28}}

\put(35,25){\makebox(5,5){\bf 25}} \put(30,30){\makebox(5,5){\bf
26}} \put(25,35){\makebox(5,5){\bf 27}}

\end{picture}
\end{scriptsize}

\centering \setlength{\unitlength}{.65mm}
\begin{scriptsize}
\begin{picture}(60,40)(-30,-20)
\put(5,20){\framebox(35,5){}} \put(10,15){\framebox(25,15){}}
\put(15,10){\framebox(15,25){}} \put(20,5){\framebox(5,35){}}

\put(25,40){\framebox(35,5){}} \put(30,35){\framebox(25,15){}}
\put(35,30){\framebox(15,25){}} \put(40,25){\framebox(5,35){}}

\put(0,25){\framebox(5,5){}} \put(5,30){\framebox(5,5){}}
\put(10,35){\framebox(5,5){}} \put(15,40){\framebox(5,5){}}
\put(20,45){\framebox(5,5){}} \put(25,50){\framebox(5,5){}}
\put(30,55){\framebox(5,5){}} \put(35,60){\framebox(5,5){}}

\put(17.5,32.5){\circle*{2}}

\put(12.5,27.5){\circle*{2}}

\put(2.5,27.5){\circle*{2}} \put(22.5,47.5){\circle*{2}}
\put(27.5,47.5){\circle*{2}} \put(37.5,47.5){\circle*{2}}
\put(27.5,42.5){\circle*{2}} \put(47.5,42.5){\circle*{2}}
\put(42.5,32.5){\circle*{2}} \put(47.5,32.5){\circle*{2}}
\put(42.5,52.5){\circle*{2}}

\put(22.5,22.5){\circle*{2}}

\put(17.5,12.5){\circle*{2}} \put(22.5,12.5){\circle*{2}}

\put(22.5,7.5){\circle*{2}} 

\end{picture}
\end{scriptsize}

\vspace{-1.5cm}

\caption{Folding by diagonals} \label{fig:fold_diag}
\end{figure}

\end{example}

\vspace{0.3cm}

F1 and F2 were used by MacWilliams and Sloane~\cite{McSl76} to
form pseudo-random arrays. F2 was also used by
Robinson~\cite{Rob97} to fold a one-dimensional ruler into a
two-dimensional Golomb rectangle. The generalization to higher
dimensions is straight forward. F3 was used in~\cite{BEMP08a} to
obtain some synchronization patterns in $\Z^D$.

\subsection{Tiling}

Tiling is one of the most basic concepts in combinatorics. We say
that a $D$-dimensional shape $\cS$ tiles the $D$-dimensional space
$\Z^D$ if disjoint copies of $\cS$ cover $\Z^D$.

\begin{remark}
We assume that our shape $\cS$ is a discrete shape, i.e., it
consists of discrete points of $\Z^D$ such that there is a path
between any two points of $\cS$ which consists only from points of
$\cS$. The shape $\cS$ in $\Z^D$ is usually not represented as a
union of points in $\Z^D$, but rather as a union of units cubes in
$\R^D$ with $2^D$ vertices in $\Z^D$. Let $A$ be the set of points
in the first representation. The set of unit cube by the second
representation is
$$\{ \cU_{(i_1, i_2 , \ldots , i_D )} ~:~ (i_1, i_2
, \ldots , i_D ) \in A  \}~,
$$
where
$$\cU_{(i_1, i_2 , \ldots , i_D )} = \{ (i_1, i_2 , \ldots , i_D )
+ \xi_1 \epsilon_1 + \xi_2 \epsilon_2 + \cdots + \xi_D \epsilon_D
~:$$ $$ 0 \leq \xi_i < 1,~ 1 \leq i \leq D \}~,
$$
and $\epsilon_i$ is a vector of length $D$ and weight one with a
{\it one} in the $i$th position. We omit the case of shapes in
$\R^D$ which are not of interest to our discussion.
\end{remark}

\vspace{0.2cm}

This cover of $\Z^D$ with disjoint copies of $\cS$ is called {\it
tiling} of $\Z^D$ with $\cS$. For each shape $\cS$ we distinguish
one of the points of $\cS$ to be the {\it center} of $\cS$. Each
copy of $\cS$ in a tiling has the center in the same related
point. The set $\cT$ of centers in a tiling defines the tiling,
and hence the tiling is denoted by the pair $(\cT,\cS)$. Given a
tiling $(\cT,\cS)$ and a grid point $(i_1,i_2,\ldots,i_D)$ we
denote by $c(i_1,i_2,\ldots,i_D)$ the center of the copy of $\cS$
for which $(i_1,i_2,\ldots,i_D) \in \cS$. We will also assume that
the origin is a center of some copy of $\cS$.
\begin{remark}
It is easy to verify that any point of $\cS$ can serve as the
center of $\cS$. If $(\cT,\cS)$ is a tiling then we can choose any
point of $\cS$ to serve as a center without affecting the fact
that $(\cT,\cS)$ is a tiling.
\end{remark}
\begin{lemma}
\label{lem:center} If $(\cT,\cS)$ is a tiling then for any given
point $(i_1,i_2,\ldots,i_D) \in \Z^D$ the point
$(i_1,i_2,\ldots,i_D)-c(i_1,i_2,\ldots,i_D)$ belongs to the shape
$\cS$ whose center is in the origin.
\end{lemma}
\begin{proof}
Let $\cS_1$ be the copy of $\cS$ whose center is in the origin and
$\cS_2$ be the copy of $\cS$ with the point
$(i_1,i_2,\ldots,i_D)$. Let $(x_1,x_2,\ldots,x_D)$ be the point in
$\cS_1$ related to the point $(i_1,i_2,\ldots,i_D)$ in $\cS_2$. By
definition, $(i_1,i_2,\ldots,i_D) = c(i_1,i_2,\ldots,i_D) +
(x_1,x_2,\ldots,x_D)$ and the lemma follows.
\end{proof}

One of the most common types of tiling is a {\it lattice tiling}.
A {\it lattice} $\Lambda$ is a discrete, additive subgroup of the
real $D$-space $\R^D$. W.l.o.g., we can assume that
\begin{equation}
\label{eq:lattice_def} \Lambda = \{ u_1 v_1 + u_2v_2 + \cdots +
u_D v_D ~:~ u_1, \ldots,u_D \in \Z \}~,
\end{equation}
where $\{ v_1,v_2,\ldots,v_D \}$ is a set of linearly independent
vectors in $\R^D$. A lattice $\Lambda$ defined by
(\ref{eq:lattice_def}) is a sublattice of $\Z^D$ if and only if
$\{ v_1,v_2,\ldots,v_D \} \subset \Z^D$. We will be interested
solely in sublattices of $\Z^D$ since our shapes are defined in
$\Z^D$. The vectors $v_1,v_2,\ldots,v_D$ are called a {\it base}
for $\Lambda \subseteq \Z^D$, and the $D \times D$ matrix
$$
{\bf G}=\left[\begin{array}{cccc}
v_{11} & v_{12} & \ldots & v_{1D} \\
v_{21} & v_{22} & \ldots & v_{2D} \\
\vdots & \vdots & \ddots & \vdots\\
v_{D1} & v_{D2} & \ldots & v_{DD} \end{array}\right]
$$
having these vectors as its rows is said to be a {\it generator
matrix} for $\Lambda$.

The {\it volume} of a lattice $\Lambda$, denoted $V( \Lambda )$,
is inversely proportional to the number of lattice points per unit
volume. More precisely, $V( \Lambda )$ may be defined as the
volume of the {\it fundamental parallelogram} $\Pi(\Lambda)$ in
$\R^D$, which is given by
$$
\Pi(\Lambda) \deff\ \{ \xi_1 v_1  + \xi_2 v_2 +  \cdots + \xi_D
v_D :  0 \leq \xi_i < 1, ~ ,1 \leq i \leq D \}.
$$
There is a simple expression for the volume of $\Lambda$, namely,
$V(\Lambda)=| \det {\bf G} |$.

We say that $\Lambda$ is a {\it lattice tiling} for $\cS$ if the
lattice points can be taken as the set $\cT$ to form a tiling
$(\cT,\cS)$. In this case we have that $|\cS|=V(\Lambda)=| \det
{\bf G} |$.

There is a large variety of literature about tiling and lattices.
We will refer the reader to two of the most interesting and
comprehensive books~\cite{CoSl93,StSz94}.

\begin{remark}
Note, that different generator matrices for the same lattice will
result in different fundamental parallelograms. This is related to
the fact that the same lattice can induce a tiling for different
shapes with the same volume. A fundamental parallelogram is always
a shape in $\R^D$ which is tiled by $\Lambda$ (usually this is not
a shape in $\Z^D$ and as a consequence, most and usually all, of
the shapes in $\Z^D$ are not fundamental parallelograms).
\end{remark}


Lattices are very fundamental structures in various coding
problems, e.g.~\cite{UrRi98,ViBo99,TVZ99} is a small sample which
does not mean to be representative. They are also applied in
multidimensional coding, e.g.~\cite{BBV}. This paper exhibits a
new application of lattices for multidimensional coding and for
discrete geometry problems.

\begin{lemma}
Let $\Lambda$ be a $D$-dimensional lattice, with a generator
matrix ${\bf G}$, and $\cS$ be a $D$-dimensional shape with a
point at the origin. $\Lambda$ is a lattice tiling for $\cS$ if
and only if $| \det {\bf G} | = |\cS|$ and there are no two points
$(i_1,i_2,\ldots,i_D)$ and $(j_1,j_2,\ldots,j_D)$ in $\cS$ such
that $(i_1-j_1,i_2-j_2,\ldots,i_D-j_D)$ is a lattice point.
\end{lemma}
\begin{proof}
Assume first that $\Lambda$ is a lattice tiling for $\cS$. The
condition on the volume of $\cS$ is trivial. Assume the contrary
that $(i_1,i_2,\ldots,i_D)$ and $(j_1,j_2,\ldots,j_D)$ are in the
copy of $\cS$, whose center is in the origin, and
$(i_1-j_1,i_2-j_2,\ldots,i_D-j_D)$ is a lattice point. It follows
that the point $(i_1,i_2,\ldots,i_D)$ is contained in the shape
centered in the origin and also in the shape centered at
$(i_1-j_1,i_2-j_2,\ldots,i_D-j_D)$, a contradiction to the fact
that $\Lambda$ is a lattice tiling for $\cS$.

Now, assume that $| \det {\bf G} | = |\cS|$ and there are no two
points $(i_1,i_2,\ldots,i_D)$ and $(j_1,j_2,\ldots,j_D)$ in $\cS$
such that $(i_1-j_1,i_2-j_2,\ldots,i_D-j_D)$ is a lattice point.
We choose the point of $\cS$ which is in the origin to be the
center of $\cS$ and we place copies of $\cS$ on each lattice point
such that the center coincide with the lattice point. Since $|
\det {\bf G} | = |\cS|$ we only have to show that there is no
point which is contained in two different copies of $\cS$ in order
to complete the proof that $\Lambda$ is a lattice tiling for
$\cS$. Assume the contrary that the point $P$ is contained in two
copies of $\cS$ with centers at $C_1$ and $C_2$. Similarly to the
proof of Lemma~\ref{lem:center} it can be shown that $P-C_1$ and
$P-C_2$ are points in the copy of $\cS$ centered at the origin,
But, $P-C_1 -(P-C_2) =C_2 - C_1$ is a lattice point (since it is a
difference of two lattice points), a contradiction to the
assumption. Hence, $\Lambda$ is a lattice tiling for $\cS$.
\end{proof}
\begin{cor}
Let $\Lambda$ be a $D$-dimensional lattice, with a generator
matrix ${\bf G}$, and $\cS$ be a $D$-dimensional shape. $\Lambda$
is a lattice tiling for $\cS$ if and only if $| \det {\bf G} | =
|\cS|$ and there are no two points $(i_1,i_2,\ldots,i_D)$ and
$(j_1,j_2,\ldots,j_D)$ in any copy of $\cS$ such that
$(i_1-j_1,i_2-j_2,\ldots,i_D-j_D)$ is a lattice point.
\end{cor}

\section{The Generalized Folding Method}
\label{sec:folding}

In this section we will generalize the definition of folding. All
the previous three definitions ({\bf F1}, {\bf F2}, and {\bf F3})
are special cases of the new definition. The new definition
involves a lattice tiling $\Lambda$, for a shape $\cS$ on which
the folding is performed.

A {\it ternary vector} of length $D$, $(d_1,d_2,\ldots,d_D)$, is a
word of length $D$, where $d_i \in \{ -1,0,+1 \}$.

Let $\cS$ be a $D$-dimensional shape and let
$\delta=(d_1,d_2,\ldots,d_D)$ be a nonzero ternary vector of
length $D$. Let $\Lambda$ be a lattice tiling for a shape $\cS$,
and let $\cS_1$ be the copy of $\cS$ which includes the origin. We
define recursively a {\it folded-row} starting in the origin. If
the point $(i_1,i_2, \ldots ,i_D )$ is the current point of
$\cS_1$ in the folded-row, then the next point on its folded-row
is defined as follows:
\begin{itemize}
\item If the point $(i_1+d_1,i_2+d_2,\ldots,i_D+d_D)$ is in
$\cS_1$ then it is the next point on the folded-row.

\item If the point $(i_1+d_1,i_2+d_2,\ldots,i_D+d_D)$ is in $\cS_2
\neq \cS_1$ whose center is in the point $(c_1,c_2,\ldots,c_D)$
then $(i_1+d_1-c_1,i_2+d_2-c_2,\ldots,i_D+d_D-c_D)$ is the next
point on the folded-row (this is a point in $\cS_1$ by
Lemma~\ref{lem:center}).
\end{itemize}

The new definition of folding is based on a lattice $\Lambda$, a
shape $\cS$, and a {\it direction} $\delta$. The triple
$(\Lambda,\cS,\delta)$ defines a {\it folding} if the definition
yields a folded-row which includes all the elements of $\cS$. It
will be proved that only $\Lambda$ and $\delta$ determine whether
the triple $(\Lambda,\cS,\delta)$ defines a folding. The role of
$\cS$ is only in the order of the elements in the folded-row; and
of course $\Lambda$ must define a lattice tiling for $\cS$.
Different lattice tilings for the same shape $\cS$ can function
completely different. Also, not all directions for the same
lattice tiling of the shape $\cS$ should define (or not define) a
folding.

\begin{remark}
It is not difficult to see that the three folding defined earlier
(F1, F2, and F3) are special cases of the new definition. The
definition of the generator matrices for the three corresponding
lattices are left as an exercise to the interested reader.
\end{remark}

\begin{remark}
The definition of ternary vectors for the direction, in which the
folding is performed, is given to guarantee that two consecutive
elements in the folded-row, are also adjacent (possibly
cyclically) in the shape $\cS$.
\end{remark}

\begin{example}
Let $\cS$ be a $2 \times 2$ square. Let $\Lambda_1$ be the lattice
whose generator matrix given by the matrix
$$
G_1=\left[\begin{array}{cc}
2 & 2 \\
0 & 2
\end{array}\right]~.
$$
$\Lambda_1$ defines a lattice tiling for $\cS$. None of the four
possible ternary vectors of length 2 define a folding with
$\Lambda$ (and $\cS$).

Let $\Lambda_2$ be the lattice whose generator matrix given by the
matrix
$$
G_2=\left[\begin{array}{cc}
2 & 1 \\
0 & 2
\end{array}\right]~.
$$
$\Lambda_2$ also defines a lattice tiling for $\cS$. Each one of
the directions $(+1,0)$, $(+1,+1)$, and $(+1,-1)$ defines a
folding with $\Lambda$ (and $\cS$). Only the direction $(0,+1)$
does not define a folding with $\Lambda$ (and $\cS$).
\end{example}

\vspace{0.5cm}

How many different folded-rows do we have? In other words, how
many different folding operations are defined in this way? There
are $3^D-1$ non-zero ternary vectors. If $\Lambda$ with the
ternary vector $(d_1,d_2,\ldots,d_D)$ define a folding then also
$\Lambda$ with the ternary vector $(-d_1,-d_2,\ldots,-d_D)$ define
a folding. The two folded-rows are in reverse order, and hence
they will be considered to be {\it equivalent}. If two folded-rows
are not equal and not a reverse pair then they will considered to
be {\it nonequivalent}. The question whether for each $D$, there
exists a $D$-dimensional shape $\cS$ with $\frac{3^D-1}{2}$
different folded-rows will be partially answered in the sequel.
Meanwhile, we present an example for $D=2$.

Before the example we shall define how we fold a sequence into a
shape $\cS$. Let $\Lambda$ be a lattice tiling for the shape $\cS$
for which $n=|\cS|$. Let $\delta$ be a direction for which
$(\Lambda,\cS,\delta)$ defines a folding. Let $\cB=b_0 b_1 \ldots
b_{n-1}$ be a sequence of length $n$. The folding of $\cB$ induced
by $(\Lambda,\cS,\delta)$ is denoted by $(\Lambda,\cS,\delta,\cB)$
and defined as the shape $\cS$ with the elements of $\cB$, where
$b_i$ is in the $i$th entry of the folded-row in $\cS$ defined by
$(\Lambda,\cS,\delta)$.

\begin{example}
\label{exm:fold1} Let $\Lambda$ be the lattice whose generator
matrix given by the matrix
$$
G=\left[\begin{array}{cc}
3 & 2 \\
7 & 1
\end{array}\right]~.
$$
One can verify that shapes tiled by this lattice have different
folded-rows. It can be proved that this is the lattice with the
smallest volume which has this property, i.e., that the four
folded-rows are different.

\vspace{0.2cm}

If our shape $\cS$ is an $1 \times 11$ array then the folding of a
sequence with length 11 is defined as follows  (the position
labelled with an $i$ is the place of the $i$th element of the
sequence).

\noindent For the direction vector $(+1,0)$ the order is given by

\vspace{0.2cm}

\begin{center}
$\begin{array}{|c|c|c|c|c|c|c|c|c|c|c|c|c|} \hline 0&1&2
&3&4&5&6&7&8&9&10 \\
\hline
\end{array}~.$
\end{center}

\vspace{0.3cm}

\noindent For the direction vector $(0,+1)$ the order is given by

\vspace{0.2cm}

\begin{center}
$\begin{array}{|c|c|c|c|c|c|c|c|c|c|c|c|c|} \hline 0&3&6
&9&1&4&7&10&2&5&8 \\
\hline
\end{array}~.$
\end{center}

\vspace{0.3cm}

\noindent For the direction vector $(+1,+1)$ the order is given by

\vspace{0.2cm}

\begin{center}
$\begin{array}{|c|c|c|c|c|c|c|c|c|c|c|c|c|} \hline 0&9&7
&5&3&1&10&8&6&4&2 \\
\hline
\end{array}~.$
\end{center}

\vspace{0.3cm}

\noindent For the direction vector $(+1,-1)$ the order is given by

\vspace{0.2cm}

\begin{center}
$\begin{array}{|c|c|c|c|c|c|c|c|c|c|c|c|c|} \hline 0&7&3
&10&6&2&9&5&1&8&4 \\
\hline
\end{array}~.$
\end{center}

\vspace{0.5cm}

\noindent If our shape $\cS$ is given by

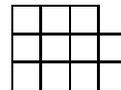
\setlength{\unitlength}{.75mm}
\begin{picture}(40,25)(10,-5)
\linethickness{.5 pt} \put(30,0){\framebox(20,10){}}
\put(30,0){\framebox(20,5){}} \put(30,0){\framebox(10,15){}}
\put(35,0){\framebox(10,15){}}
\end{picture}

\noindent then the folding of a sequence of length 11 is depicted
in Figure~\ref{fig:fold_shape1}.

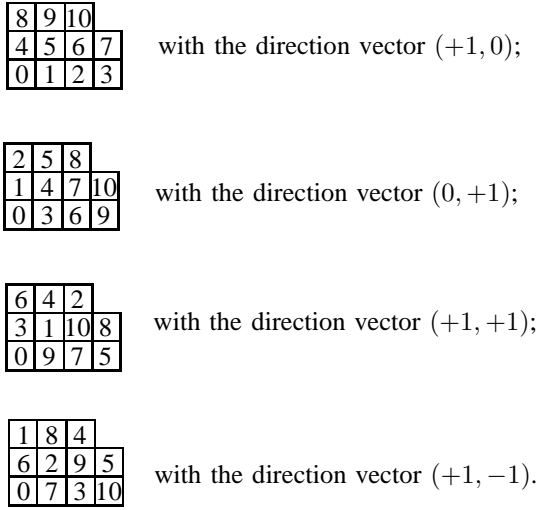
\begin{figure}[tb]

\centering \setlength{\unitlength}{.75mm}
\begin{picture}(60,40)(20,-20)
\linethickness{.5 pt}

\put(30,0){\framebox(20,10){}} \put(30,0){\framebox(20,5){}}

\put(30,0){\framebox(10,15){}} \put(35,0){\framebox(10,15){}}

\put(30,0){\makebox(5,5){0}} \put(35,0){\makebox(5,5){1}}
\put(40,0){\makebox(5,5){2}} \put(45,0){\makebox(5,5){3}}
\put(30,5){\makebox(5,5){4}} \put(35,5){\makebox(5,5){5}}
\put(40,5){\makebox(5,5){6}} \put(45,5){\makebox(5,5){7}}
\put(30,10){\makebox(5,5){8}} \put(35,10){\makebox(5,5){9}}
\put(40,10){\makebox(5,5){10}}

\end{picture}
\begin{picture}(50,40)(25,-26)
with the direction vector $(+1,0)$;
\end{picture}

\centering \setlength{\unitlength}{.75mm}
\begin{picture}(57,40)(27,-35)
\linethickness{.5 pt}

\put(30,0){\framebox(20,10){}} \put(30,0){\framebox(20,5){}}

\put(30,0){\framebox(10,15){}} \put(35,0){\framebox(10,15){}}

\put(30,0){\makebox(5,5){0}} \put(35,0){\makebox(5,5){3}}
\put(40,0){\makebox(5,5){6}} \put(45,0){\makebox(5,5){9}}
\put(30,5){\makebox(5,5){1}} \put(35,5){\makebox(5,5){4}}
\put(40,5){\makebox(5,5){7}} \put(45,5){\makebox(5,5){10}}
\put(30,10){\makebox(5,5){2}} \put(35,10){\makebox(5,5){5}}
\put(40,10){\makebox(5,5){8}}

\end{picture}
\begin{picture}(40,40)(29,-40)
with the direction vector $(0,+1)$;
\end{picture}

\centering \setlength{\unitlength}{.75mm}
\begin{picture}(54,40)(49,-50)
\linethickness{.5 pt}

\put(30,0){\framebox(20,10){}} \put(30,0){\framebox(20,5){}}

\put(30,0){\framebox(10,15){}} \put(35,0){\framebox(10,15){}}

\put(30,0){\makebox(5,5){0}} \put(35,0){\makebox(5,5){9}}
\put(40,0){\makebox(5,5){7}} \put(45,0){\makebox(5,5){5}}
\put(30,5){\makebox(5,5){3}} \put(35,5){\makebox(5,5){1}}
\put(40,5){\makebox(5,5){10}} \put(45,5){\makebox(5,5){8}}
\put(30,10){\makebox(5,5){6}} \put(35,10){\makebox(5,5){4}}
\put(40,10){\makebox(5,5){2}}

\end{picture}

\begin{picture}(30,40)(5,-97)
with the direction vector $(+1,+1)$;
\end{picture}

\centering \setlength{\unitlength}{.75mm}
\begin{picture}(51,40)(50,-106)
\linethickness{.5 pt}

\put(30,0){\framebox(20,10){}} \put(30,0){\framebox(20,5){}}

\put(30,0){\framebox(10,15){}} \put(35,0){\framebox(10,15){}}

\put(30,0){\makebox(5,5){0}} \put(35,0){\makebox(5,5){7}}
\put(40,0){\makebox(5,5){3}} \put(45,0){\makebox(5,5){10}}
\put(30,5){\makebox(5,5){6}} \put(35,5){\makebox(5,5){2}}
\put(40,5){\makebox(5,5){9}} \put(45,5){\makebox(5,5){5}}
\put(30,10){\makebox(5,5){1}} \put(35,10){\makebox(5,5){8}}
\put(40,10){\makebox(5,5){4}}
\end{picture}

\begin{picture}(30,40)(4,-150)
with the direction vector $(+1,-1)$.
\end{picture}
\vspace{-11cm}

\caption{Folding of the first shape} \label{fig:fold_shape1}
\end{figure}


\noindent Finally, if our shape $\cS$ is given by

\vspace{0.5cm}

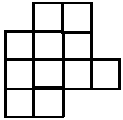
\setlength{\unitlength}{.75mm}
\begin{picture}(40,25)(10,-5)
\linethickness{.5 pt}

\put(30,0){\framebox(10,15){}} \put(30,5){\framebox(20,5){}}
\put(35,5){\framebox(10,15){}} \put(35,0){\framebox(5,5){}}
\put(40.1,15){\framebox(5,5){}}
\end{picture}


\noindent then the folding of a sequence of length 11 is depicted
in Figure~\ref{fig:fold_shape2}.

\begin{figure}[tb]

\centering \setlength{\unitlength}{.75mm}
\begin{picture}(60,40)(20,-20)
\linethickness{.5 pt}

\put(30,0){\framebox(10,15){}} \put(30,5){\framebox(20,5){}}
\put(35,5){\framebox(10,15){}} \put(35,0){\framebox(5,5){}}
\put(40.1,15){\framebox(5,5){}}

\put(30,0){\makebox(5,5){0}} \put(35,0){\makebox(5,5){1}}
\put(35,15){\makebox(5,5){2}} \put(40,15){\makebox(5,5){3}}
\put(30,5){\makebox(5,5){4}} \put(35,5){\makebox(5,5){5}}
\put(40,5){\makebox(5,5){6}} \put(45,5){\makebox(5,5){7}}
\put(30,10){\makebox(5,5){8}} \put(35,10){\makebox(5,5){9}}
\put(40,10){\makebox(5,5){10}}
\end{picture}
\begin{picture}(50,40)(25,-26)
with the direction vector $(+1,0)$;
\end{picture}

\centering \setlength{\unitlength}{.75mm}
\begin{picture}(57,40)(27,-35)
\linethickness{.5 pt}

\put(30,0){\framebox(10,15){}} \put(30,5){\framebox(20,5){}}
\put(35,5){\framebox(10,15){}} \put(35,0){\framebox(5,5){}}
\put(40.1,15){\framebox(5,5){}}

\put(30,0){\makebox(5,5){0}} \put(35,0){\makebox(5,5){3}}
\put(35,15){\makebox(5,5){6}} \put(40,15){\makebox(5,5){9}}
\put(30,5){\makebox(5,5){1}} \put(35,5){\makebox(5,5){4}}
\put(40,5){\makebox(5,5){7}} \put(45,5){\makebox(5,5){10}}
\put(30,10){\makebox(5,5){2}} \put(35,10){\makebox(5,5){5}}
\put(40,10){\makebox(5,5){8}}
\end{picture}
\begin{picture}(40,40)(29,-40)
with the direction vector $(0,+1)$;
\end{picture}

\centering \setlength{\unitlength}{.75mm}
\begin{picture}(54,40)(49,-50)
\linethickness{.5 pt}

\put(30,0){\framebox(10,15){}} \put(30,5){\framebox(20,5){}}
\put(35,5){\framebox(10,15){}} \put(35,0){\framebox(5,5){}}
\put(40.1,15){\framebox(5,5){}}

\put(30,0){\makebox(5,5){0}} \put(35,0){\makebox(5,5){9}}
\put(35,15){\makebox(5,5){7}} \put(40,15){\makebox(5,5){5}}
\put(30,5){\makebox(5,5){3}} \put(35,5){\makebox(5,5){1}}
\put(40,5){\makebox(5,5){10}} \put(45,5){\makebox(5,5){8}}
\put(30,10){\makebox(5,5){6}} \put(35,10){\makebox(5,5){4}}
\put(40,10){\makebox(5,5){2}}
\end{picture}

\begin{picture}(30,40)(5,-97)
with the direction vector $(+1,+1)$;
\end{picture}

\centering \setlength{\unitlength}{.75mm}
\begin{picture}(51,40)(50,-106)
\linethickness{.5 pt}

\put(30,0){\framebox(10,15){}} \put(30,5){\framebox(20,5){}}
\put(35,5){\framebox(10,15){}} \put(35,0){\framebox(5,5){}}
\put(40.1,15){\framebox(5,5){}}

\put(30,0){\makebox(5,5){0}} \put(35,0){\makebox(5,5){7}}
\put(35,15){\makebox(5,5){3}} \put(40,15){\makebox(5,5){10}}
\put(30,5){\makebox(5,5){6}} \put(35,5){\makebox(5,5){2}}
\put(40,5){\makebox(5,5){9}} \put(45,5){\makebox(5,5){5}}
\put(30,10){\makebox(5,5){1}} \put(35,10){\makebox(5,5){8}}
\put(40,10){\makebox(5,5){4}}
\end{picture}

\begin{picture}(30,40)(4,-150)
with the direction vector $(+1,-1)$.
\end{picture}
\vspace{-11cm}

\caption{Folding of the second shape} \label{fig:fold_shape2}
\end{figure}
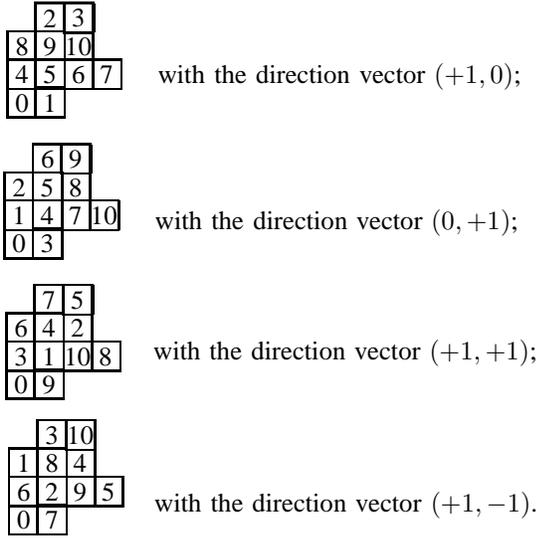

\end{example}

Next, we aim to find sufficient and necessary conditions that a
triple $(\Lambda,\cS,\delta)$ defines a folding. We start with a
simple characterization for the order of the elements in a
folded-row.

\begin{lemma}
\label{lem:order_fold} Let $\Lambda$ be a lattice tiling for the
shape $\cS$ and let $\delta = (d_1,d_2,\ldots,d_D)$ be a nonzero
ternary vector. Let $g(i)=(i \cdot d_1,\ldots,i \cdot d_D)-c(i
\cdot d_1,\ldots,i \cdot d_D)$ and let $i_1$, $i_2$ be two
integers. Then $g(i_1)=g(i_2)$ if and only if $g(i_1+1)=g(i_2+1)$.
\end{lemma}
\begin{proof}
The lemma follows immediately from the observation that
$g(i_1)=g(i_2)$ if and only if $(i_1 \cdot d_1,\ldots,i_1 \cdot
d_D)$ and $(i_2 \cdot d_1,\ldots,i_2 \cdot d_D)$ are the same
related position in $\cS$, i.e., corresponds to the same position
of the folded-row.
\end{proof}

The next two lemmas are an immediate consequence of the
definitions and provide a concise condition whether the triple
$(\Lambda,\cS,\delta)$ defines a folding.

\begin{lemma}
\label{lem:tile_fold1} Let $\Lambda$ be a lattice tiling for the
shape $\cS$ and let $\delta = (d_1,d_2,\ldots,d_D)$ be a nonzero
ternary vector. $(\Lambda,\cS,\delta)$ defines a folding if and
only if the set $\{ (i \cdot d_1,i \cdot d_2,\ldots,i \cdot
d_D)-c(i \cdot d_1,i \cdot d_2,\ldots,i \cdot d_D) ~:~ 0 \leq i <
| \cS | \}$ contains $| \cS |$ distinct elements.
\end{lemma}
\begin{proof}
The lemma is an immediate consequence of
Lemmas~\ref{lem:center},~\ref{lem:order_fold}, and the definition
of folding.
\end{proof}

\begin{lemma}
\label{lem:tile_fold2} Let $\Lambda$ be a lattice tiling for the
shape $\cS$ and let $\delta = (d_1,d_2,\ldots,d_D)$ be a nonzero
ternary vector. $(\Lambda,\cS,\delta)$ defines a folding if and
only if $(|\cS| \cdot d_1,\ldots,|\cS| \cdot d_D)-c(|\cS| \cdot
d_1,\ldots,|\cS| \cdot d_D)=(0,\ldots,0)$ and for each $i$, $0 < i
< |\cS|$ we have $(i \cdot d_1,\ldots,i \cdot d_D)-c(i \cdot
d_1,\ldots,i \cdot d_D) \neq (0, \ldots ,0)$.
\end{lemma}
\begin{proof}
Assume first that $(\Lambda,\cS,\delta)$ defines a folding. If for
some $0 < j < |\cS|$ we have $(j \cdot d_1,\ldots,j \cdot d_D)-c(j
\cdot d_1,\ldots,j\cdot d_D) = (0, \ldots ,0)$ then $g(j)=g(0)$
and hence by Lemma~\ref{lem:order_fold} the folded-row will have
at most $j$ elements of $\cS$. Since $j < |\cS|$ we will have that
$(\Lambda,\cS,\delta)$ does not define a folding. On the other
hand, Lemma~\ref{lem:order_fold} also implies that if
$(\Lambda,\cS,\delta)$ defines a folding then $g(|\cS|)=
(0,\ldots,0)$.

Now assume that $(|\cS| \cdot d_1,\ldots,|\cS| \cdot d_D)-c(|\cS|
\cdot d_1,\ldots,|\cS| \cdot d_D)=(0,\ldots,0)$ and for each $i$,
$0 < i < |\cS|$ we have $(i \cdot d_1,\ldots,i \cdot d_D)-c(i
\cdot d_1,\ldots,i \cdot d_D) \neq (0, \ldots ,0)$. Let $0 < i_1 <
i_2 < |\cS|$; if $g(i_1) = g(i_2)$ then by
Lemma~\ref{lem:order_fold} we have $g(i_2-i_1)=g(0)=(0, \ldots
,0)$, a contradiction. Therefore, the folded-row contains all the
elements of $\cS$ and hence by definition $(\Lambda,\cS,\delta)$
defines a folding.
\end{proof}
\begin{cor}
\label{cor:lattice_points} If $(\Lambda,\cS,\delta)$, $\delta =
(d_1,d_2,\ldots,d_D)$, defines a folding then the point $(|\cS|
\cdot d_1,\ldots,|\cS| \cdot d_D)$ is a lattice point.
\end{cor}

\vspace{0.5cm}

Before considering the general $D$-dimensional case we want to
give a simple condition to check whether the triple
$(\Lambda,\cS,\delta)$ defines a folding in the two-dimensional
case. For each one of the four possible ternary vector we will
give a necessary and sufficient condition that the triple
$(\Lambda,\cS,\delta)$ defines a folding.

\begin{lemma}
\label{lem:det_points} Let $G$ be the generator matrix of a
lattice $\Lambda$ and let $s=| \det G |$. Then the points $(0,s)$,
$(s,0)$, $(s,s)$, and $(s,-s)$ are lattice points.
\end{lemma}
\begin{proof}
It is sufficient to prove that the points $(0,s)$, $(s,0)$ are
lattice points. Let $\Lambda$ be a lattice whose generator matrix
is given by
$$
G=\left[\begin{array}{cc}
v_{11} & v_{12} \\
v_{21} & v_{22}
\end{array}\right]~.
$$
W.l.o.g. we assume that $| \det G | > 0$, i.e., $s= v_{11} v_{22}
- v_{12} v_{21}$. Since $v_{22} ( v_{11} , v_{12} ) - v_{12} (
v_{21} , v_{22} ) = (s,0)$ and $v_{11} ( v_{21} , v_{22} ) -
v_{21} ( v_{11} , v_{12} ) = (0,s)$, it follows that $(0,s)$,
$(s,0)$ are lattice points.
\end{proof}

\begin{theorem}
\label{thm:cond_fold2D} Let $\Lambda$ be a lattice whose generator
matrix is given by
$$
G=\left[\begin{array}{cc}
v_{11} & v_{12} \\
v_{21} & v_{22}
\end{array}\right]~.
$$
If $\Lambda$ defines a lattice tiling for the shape $\cS$ then the
triple $(\Lambda,\cS,\delta)$ defines a folding

\begin{itemize}
\item with the ternary vector $\delta =(+1,+1)$ if and only if
$\text{g.c.d.}(v_{22}-v_{21},v_{11}-v_{12})=1$;

\item with the ternary vector $\delta =(+1,-1)$ if and only if
$\text{g.c.d.}(v_{22}+v_{21},v_{11}+v_{12})=1$;

\item with the ternary vector $\delta =(+1,0)$ if and only if
$\text{g.c.d.}(v_{12},v_{22})=1$;

\item with the ternary vector $\delta =(0,+1)$ if and only if
$\text{g.c.d.}(v_{11},v_{21})=1$.
\end{itemize}
\end{theorem}
\begin{proof}
We will prove the case where $\delta =(+1,+1)$; the other three
cases are proved similarly.

Let $\Lambda$ be a lattice tiling for the shape $\cS$. By
Lemma~\ref{lem:det_points} we have that $(|\cS|,|\cS|)$ is a
lattice point. Therefore, there exist two integers $\alpha_1$ and
$\alpha_2$ such that $\alpha_1 (v_{11},v_{12}) + \alpha_2
(v_{21},v_{22}) =(|\cS|,|\cS|)$, i.e., $\alpha_1 v_{11} + \alpha_2
v_{21} = \alpha_1 v_{12} + \alpha_2 v_{22}= | \cS |= v_{11} v_{22}
- v_{12} v_{21}$. These equations have exactly one solution,
$\alpha_1 = v_{22}-v_{21}$ and $\alpha_2 = v_{11}-v_{12}$.
By Lemma~\ref{lem:tile_fold2}, $(\Lambda,\cS,\delta)$ defines a
folding if and only if $(|\cS|,|\cS|)=c(|\cS|,|\cS|)$ and for each
$i$, $0 < i < |\cS|$ we have $(i,i) \neq c(i,i)$.

Assume first that $\text{g.c.d.}(v_{22}-v_{21},v_{11}-v_{12})=1$.
Assume for the contrary, that there exist three integers $i$,
$\beta_1$, and $\beta_2$, such that $\beta_1 (v_{11},v_{12}) +
\beta_2 (v_{21},v_{22}) =(i,i)$, $0 < i < |\cS|$. Hence, $\beta_1
v_{11} + \beta_2 v_{21} = \beta_1 v_{12} + \beta_2 v_{22} = i$,
i.e., $\frac{\beta_2}{\beta_1} =
\frac{v_{11}-v_{12}}{v_{22}-v_{21}}=\frac{\alpha_2}{\alpha_1}$.
Since $\alpha_1 = v_{22}-v_{21}$ and
$\text{g.c.d.}(v_{22}-v_{21},v_{11}-v_{12})=1$, it follows that
$\beta_1 = \gamma \alpha_1$ and $\beta_2 = \gamma \alpha_2$, for
some integer $\gamma > 0$ (w.l.o.g. we can assume tha $\gamma >
0$). Therefore, $i=\beta_1 v_{11} + \beta_2 v_{21} = \gamma
\alpha_1 v_{11} + \gamma \alpha_2 v_{21} = \gamma |S| \geq |\cS|$,
a contradiction. Thus, it follows from Lemma~\ref{lem:tile_fold2}
that if $\text{g.c.d.}(v_{22}-v_{21},v_{11}-v_{12})=1$ then
$(\Lambda,\cS,\delta)$ defines a folding with the ternary vector
$\delta =(+1,+1)$.

Assume now that $(\Lambda,\cS,\delta)$ defines a folding with the
ternary vector $\delta =(+1,+1)$. Assume for the contrary that
$\text{g.c.d.}(v_{22}-v_{21},v_{11}-v_{12})=\nu>1$. Since
$\text{g.c.d.}(v_{22}-v_{21},v_{11}-v_{12})=\nu>1$, it follows
that $\beta_1 = \frac{v_{22}-v_{21}}{\nu}$ and $\beta_2 =
\frac{v_{11}-v_{12}}{\nu}$ are integers. Therefore, $\beta_1
(v_{11},v_{12}) + \beta_2 (v_{21},v_{22})
=(\frac{|\cS|}{\nu},\frac{|\cS|}{\nu})$ and as a consequence
$\frac{|\cS|}{\nu}$ is an integer. Hence, by
Lemma~\ref{lem:tile_fold2} we have that $(\Lambda,\cS,\delta)$
does not define a folding, a contradiction. Thus, if
$(\Lambda,\cS,\delta)$ defines a folding with the ternary vector
$\delta =(+1,+1)$ then
$\text{g.c.d.}(v_{22}-v_{21},v_{11}-v_{12})=1$.
\end{proof}

Theorem~\ref{thm:cond_fold2D} is generalized for the
$D$-dimensional. This generalization will be presented in
Theorem~\ref{thm:new_general_condition} given in Appendix A.

There are cases when we can determine immediately without going
into all the computation, whether $(\Lambda,\cS,\delta)$ defines a
folding. It will be a consequence of the following lemmas.

\begin{lemma}
\label{lem:size_fold} $~$
\begin{itemize}
\item The number of elements in a folded-row does not depend on
the point of $\cS$ chosen to be the center of $\cS$.

\item The number of elements in a folded-row is a divisor of
$|\cS|$, i.e., a divisor of $V( \Lambda )$.
\end{itemize}
\end{lemma}
\begin{proof}
By Lemmas~\ref{lem:order_fold} and~\ref{lem:tile_fold2} and the
definition of the folded-row, if we start the folded-row in the
origin then the number of elements in the folded-row is the
smallest $t$ such that $t \cdot \delta$ is a lattice point (since
the folded-row starts at a lattice point and ends one step before
it reaches again a lattice point). This implies that the number of
elements in a folded-row does not depend on the point of $\cS$
chosen to be the center of $\cS$. We can make any point of $\cS$
to be the center of $\cS$ and hence any point can be at the
origin. Therefore, all folded-rows with the direction $\delta$
have $t$ elements. Any two folded-rows are either equal or
disjoint. Hence $t$ must be a divisor $|\cS|$ and $t$ does not
depend on which point of $\cS$ is the center.
\end{proof}

The next lemma is an immediate consequence from the definition of
a folded-row.

\begin{lemma}
The number of elements in a folded-row is one if and only if
$\delta$ is a lattice point.
\end{lemma}

\begin{cor}
\label{cor:prime_fold} If the volume of a lattice is a prime
number then it defines a folding with a direction $\delta$ unless
$\delta$ is a lattice point.
\end{cor}

By Theorem~\ref{thm:new_general_condition} it is clear that we can
determine whether the triple $(\Lambda,\cS,\delta)$ defines a
folding only by the lattice $\Lambda$ and the ternary direction
vector $\delta$. The role of $\cS$ is only in the fact that
$\Lambda$ should be a lattice tiling for $\cS$. But, it would be
easier to examine simpler shapes (like rectangle) than more
complicated shapes even so they have the same lattice tiling
$\Lambda$. This leads to an important tool that we will use to
find an appropriate folding for a shape $\cS'$. We will use a
folding of a simpler shape $\cS$ with the same volume and apply
iteratively the following theorem. The proof of the theorem is an
immediate consequence from the definitions of lattice tiling and
folding.

\begin{theorem}
\label{thm:trans_shape} Let $\Lambda$ be a lattice tiling for the
$D$-dimensional shape $\cS$, let $\delta=(d_1,d_2,\ldots,d_D)$ be
a nonzero ternary vector, and $(\Lambda,\cS,\delta)$ defines a
folding. Assume the origin is a point in the copy $\cS'$ of $\cS$,
$(i_1,i_2,\ldots,i_D) \in \cS'$, $(i_1+d_1,i_2+d_2,\ldots,i_D+d_D)
\in \tilde{\cS}$, $\cS' \neq \tilde{\cS}$, and the center of
$\tilde{\cS}$ is the point $(c_1,c_2,\ldots,c_D)$. Then $\Lambda$
is also a lattice tiling for the shape $\cQ=\cS' \cup
\{(i_1+d_1,i_2+d_2,\ldots,i_D+d_D)\} \setminus
\{(i_1+d_1-c_1,i_2+d_2-c_2,\ldots,i_D+d_D-c_D)\}$ and the triple
$(\Lambda,\cQ,\delta)$ also defines a folding.
\end{theorem}

\subsection{Further generalization of folding}
\label{sec:further_generalize}

So far we have used a ternary vector to indicate the direction in
which the supposed folding is performed. The use of a ternary
vector is implied by a natural requirement that consecutive
elements on the folded-row will be also consecutive elements in
the shape (up to cyclic shift). But, as we will see in the sequel,
and specifically in the application of Sections~\ref{sec:bounds}
and~\ref{sec:pseudo-random}, we don't need this requirement. This
leads for further generalization and modification of folding which
will yield a better understanding of the operation and its
properties.

A {\it direction vector} (direction in short) of length $D$,
$(d_1,d_2,\ldots,d_D)$, is a nonzero word of length $D$, where
$d_i \in \Z$. The definitions of a folded-row and folding remain
as before with the exception that instead of a nonzero ternary
vector we use any nonzero integer direction vector. Also, all the
results obtained in this section remain true with the same proofs.
The only exception is Theorem~\ref{thm:cond_fold2D} for which we
need a generalized version which will be given in the sequel.

\begin{lemma}
\label{lem:pos_direct} Let $\Lambda$ be a lattice tiling for the
shape $\cS$. Let $(d_1,d_2,\ldots,d_D)$ be a direction vector,
$(i_1,i_2,\ldots,i_D)$ be a lattice point, and the point
$(d_1,d_2,\ldots,d_D)$ is in the shape $\cS$ whose center is in
the origin. Then the folded-rows defined by the directions
$(d_1,d_2,\ldots,d_D)$ and $(i_1+d_1,i_2+d_2,\ldots,i_d+d_D)$ are
equivalent.
\end{lemma}
\begin{proof}
Follows immediately from the observation that
$c(i_1+d_1,i_2+d_2,\ldots,i_d+d_D)=(i_1,i_2,\ldots,i_D)$.
\end{proof}

In view of Lemma~\ref{lem:pos_direct} we should examine only the
$|\cS|-1$ directions related to the points of $\cS$ whose center
is in the origin. Hence, in the sequel each direction $\delta
=(d_1,d_2,\ldots,d_D)$ will have the property that the point
$(d_1,d_2,\ldots,d_D)$ will be contained in the copy of $\cS$
whose center is in the origin. One might puzzle how this relates
to the observation that the necessary and sufficient conditions
that a direction defines a folding depend only on the generator
matrix of $\Lambda$ and not on $\cS$? The answer is that the
folded-row itself is defined on the elements of $\cS$. Therefore,
$\Lambda$ will have different directions and folded-rows depending
on the shape $\cS$.

\begin{remark}
If we consider only the $|\cS|-1$ directions related to the points
of $\cS$ whose center is in the origin, some on the ternary
direction vectors might not be considered (directions which form
an equivalent folding will be considered). This is another reason
for the distinction between the definitions of direction vectors
(ternary vector and integer vector). Each definition has a
different purpose.
\end{remark}

\begin{lemma}
Let $\Lambda$ be a lattice tiling for the shape $\cS$, $n=|\cS|$.
Let $\delta =(d_1,d_2,\ldots,d_D)$ be a direction vector and let
$f_0 f_1 \ldots f_{n-1}$ be its folded-row, where $f_0 =
(0,0,\ldots,0)$ and $f_1 =(d_1,d_2,\ldots,d_D)$. Then the
direction $\delta' = f_i$ defines a folding if and only if
$\text{g.c.d.}(i,n)=1$. If the direction $\delta' = f_i$ defines a
folding then its folded-row is $f_0 f_i f_{2i} \ldots f_{n-i}$,
where indices are taken modulo $n$.
\end{lemma}
\begin{proof}
By definition and by Lemma~\ref{lem:order_fold} we have that
$\delta' = f_i = (i \cdot d_1,i \cdot d_2,\ldots,i \cdot d_D) -
c(i \cdot d_1,i \cdot d_2,\ldots,i \cdot d_D)$ and $f_{\ell \cdot
i} = (\ell \cdot i \cdot d_1, \ell \cdot i \cdot d_2,\ldots,\ell
\cdot i \cdot d_D) -c(\ell \cdot i \cdot d_1, \ell \cdot i \cdot
d_2,\ldots,\ell \cdot i \cdot d_D)$. Since the sequence $f_0 f_1
\ldots f_{n-1}$ consists of $n$ distinct points of $\Z^D$, it
follows that the sequence $f_0 f_i f_{2i} \ldots f_{n-i}$ consists
of $n$ distinct points of $\Z^D$ if and only if
$\text{g.c.d.}(i,n)=1$. Thus, the lemma follows.
\end{proof}
\begin{cor}
\label{cor:num_direct} Let $\Lambda$ be a lattice tiling for the
shape $\cS$. There exists one folding with respect to $\Lambda$ if
and only if the number of nonequivalent folding operations with
respect to $\Lambda$ is $\frac{\phi ( | \cS |)}{2}$, where $\phi (
\cdot )$ is the Euler function.
\end{cor}

Corollary~\ref{cor:num_direct} implies that once we have one
folding operation with its folded-row, then we can easily find and
compute all the other folding operations with their folded-rows.
It also implies that once the necessary and sufficient conditions
for the existence of one folding in the related theorems are
satisfied, then the necessary and sufficient conditions for the
existence of many other folding are also satisfied. Nevertheless,
Corollary~\ref{cor:num_direct} does not guarantee that there will
be a direction which defines a folding. This fact is shown in the
next example given in terms of a lemma.

\begin{lemma}
Let $\gamma$ a positive integer greater than one, $a_1$,
$a_2$,...,$a_D$, be nonzero integers, and $b_i$, $b_2$,...,$b_D$
be nonzero integers such that either $b_i=a_i$ or $b_i = a_i
\gamma$, for each $1 \leq i \leq D$, and $|\{ i ~:~ b_i=a_i
\gamma, ~ 1 \leq i \leq D \}| \geq 2$. Let $\cS$ be a
$D$-dimensional shape and $\Lambda$ be a lattice tiling for $\cS$
whose generator matrix is given by
$$ \left[\begin{array}{cccc}
b_1 & 0 & \ldots & 0 \\
0 & b_2 & \ldots & 0 \\
\vdots & \vdots & \ddots & \vdots\\
0 & 0 & \ldots & b_D  \end{array}\right]~.
$$
Then there is no direction $\delta$ for which the triple
$(\Lambda,\cS,\delta)$ defines a folding.
\end{lemma}
\begin{proof}
Let $\delta = (d_1 , d_2 , \ldots , d_D )$ be any direction vector
and let $\sigma = \gamma \prod_{i=1}^D a_i$. Then, $\sigma <
|\cS|$ and for any given shape $\cS$ for which $\Lambda$ is a
lattice tiling we have $( \sigma \cdot d_1 , \sigma \cdot d_2 ,
\ldots , \sigma \cdot d_D ) - c(\sigma \cdot d_1 , \sigma \cdot
d_2 , \ldots , \sigma \cdot d_D ) = (0,0, \ldots , 0)$. Hence, by
Lemma~\ref{lem:tile_fold2}, the triple $(\Lambda,\cS,\delta)$ does
not define a folding.
\end{proof}

\begin{lemma}
Let $\Lambda$ be a lattice tiling for the shape $\cS$. If $|\cS|$
is a prime number then there exists $\frac{|\cS|-1}{2}$ different
directions which form $\frac{|\cS|-1}{2}$ nonequivalent
folded-rows.
\end{lemma}
\begin{proof}
Let $p=|\cS|$ be a prime number. By Corollary~\ref{cor:prime_fold}
a direction $\delta$ defines a folding if a and only if $\delta$
is not a lattice point. A shape $\cS$ in the tiling contains
exactly one lattice point. Therefore, by
Corollary~\ref{cor:num_direct}, any one of the $p-1$ directions
defined by the non-lattice points of $\cS$ defines a folding.
\end{proof}

\begin{example}
Consider the lattice $\Lambda$ of Example~\ref{exm:fold1}. It is a
lattice tiling for three shapes given in Example~\ref{exm:fold1}.
For each shape, four nonequivalent folding operations are given in
Example~\ref{exm:fold1}. We will demonstrate the fifth one now.

For the $1 \times 11$ array the fifth folding operation has the
direction vector $(+2,0)$ and the order is given by

\vspace{0.2cm}

\begin{center}
$\begin{array}{|c|c|c|c|c|c|c|c|c|c|c|c|c|} \hline 0&6&1
&7&2&8&3&9&4&10&5 \\
\hline
\end{array}~.$
\end{center}

\vspace{0.2cm}

For the second shape and the direction vector $(+2,0)$, the order
is given by

\setlength{\unitlength}{.75mm}
\begin{picture}(40,25)(10,-5)
\linethickness{.5 pt}

\put(30,0){\framebox(20,10){}} \put(30,0){\framebox(20,5){}}

\put(30,0){\framebox(10,15){}} \put(35,0){\framebox(10,15){}}

\put(30,0){\makebox(5,5){0}} \put(35,0){\makebox(5,5){6}}
\put(40,0){\makebox(5,5){1}} \put(45,0){\makebox(5,5){7}}
\put(30,5){\makebox(5,5){2}} \put(35,5){\makebox(5,5){8}}
\put(40,5){\makebox(5,5){3}} \put(45,5){\makebox(5,5){9}}
\put(30,10){\makebox(5,5){4}} \put(35,10){\makebox(5,5){10}}
\put(40,10){\makebox(5,5){5}}

\end{picture}

For the third shape and the direction vector $(+1,+2)$, the order
is given by

\vspace{0.2cm}


\setlength{\unitlength}{.75mm}
\begin{picture}(40,25)(10,-5)
\linethickness{.5 pt}

\put(30,0){\framebox(10,15){}} \put(30,5){\framebox(20,5){}}
\put(35,5){\framebox(10,15){}} \put(35,0){\framebox(5,5){}}
\put(40.1,15){\framebox(5,5){}}

\put(30,0){\makebox(5,5){0}} \put(35,0){\makebox(5,5){5}}
\put(35,15){\makebox(5,5){10}} \put(40,15){\makebox(5,5){4}}
\put(30,5){\makebox(5,5){9}} \put(35,5){\makebox(5,5){3}}
\put(40,5){\makebox(5,5){8}} \put(45,5){\makebox(5,5){2}}
\put(30,10){\makebox(5,5){7}} \put(35,10){\makebox(5,5){1}}
\put(40,10){\makebox(5,5){6}}
\end{picture}

\end{example}

We continue now with the theorem which generalizes
Theorem~\ref{thm:cond_fold2D}. Indeed, it was enough to prove the
generalization only for the $D$-dimensional case. But, we feel
that making the generalizations one step at a time, first for
$D=2$ and after that for any $D \geq 2$, will make it easier on
the reader, and especially as we are using some different
reasoning in these two generalizations.
\begin{theorem}
\label{thm:new_cond_fold2D} Let $\Lambda$ be a lattice whose
generator matrix is given by
$$
G=\left[\begin{array}{cc}
v_{11} & v_{12} \\
v_{21} & v_{22}
\end{array}\right]~.
$$
Let $d_1$ and $d_2$ be two positive integers and $\tau =
\text{g.c.d.}(d_1 , d_2)$. If $\Lambda$ defines a lattice tiling
for the shape $\cS$ then the triple $(\Lambda,\cS,\delta)$ defines
a folding

\begin{itemize}
\item with the ternary vector $\delta =(+d_1,+d_2)$ if and only if
$\text{g.c.d.}(\frac{d_1 v_{22}-d_2 v_{21}}{\tau},\frac{d_2
v_{11}-d_1 v_{12}}{\tau})=1$ and $\text{g.c.d.}(\tau , | \cS
|)=1$;

\item with the ternary vector $\delta =(+d_1,-d_2)$ if and only if
$\text{g.c.d.}(\frac{d_1 v_{22}+d_2 v_{21}}{\tau},\frac{d_2
v_{11}+d_1 v_{12}}{\tau})=1$ and $\text{g.c.d.}(\tau , | \cS
|)=1$;

\item with the ternary vector $\delta =(+d_1,0)$ if and only if
$\text{g.c.d.}(v_{12},v_{22})=1$ and $\text{g.c.d.}(d_1 , | \cS
|)=1$;

\item with the ternary vector $\delta =(0,+d_2)$ if and only if
$\text{g.c.d.}(v_{11},v_{21})=1$ and $\text{g.c.d.}(d_2 , | \cS
|)=1$.
\end{itemize}
\end{theorem}
\begin{proof}
We will prove the case where $\delta =(+d_1,+d_2)$; the other
three cases are proved similarly.

Let $\Lambda$ be a lattice tiling for the shape $\cS$. By
Lemma~\ref{lem:det_points} we have that $(|\cS| \cdot d_1 ,|\cS|
\cdot d_2)$ is a lattice point. Therefore, there exist two
integers $\alpha_1$ and $\alpha_2$ such that $\alpha_1
(v_{11},v_{12}) + \alpha_2 (v_{21},v_{22}) =(|\cS| \cdot d_1
,|\cS| \cdot d_2)$, i.e., $\alpha_1 v_{11} + \alpha_2 v_{21} = d_1
|\cS|$, $\alpha_1 v_{12} + \alpha_2 v_{22}= d_2 | \cS |$, and
$|\cS|= v_{11} v_{22} - v_{12} v_{21}$. These equations have
exactly one solution, $\alpha_1 = d_1 v_{22}-d_2 v_{21}$ and
$\alpha_2 = d_2 v_{11}-d_1 v_{12}$. By Lemma~\ref{lem:tile_fold2},
$(\Lambda,\cS,\delta)$ defines a folding if and only if $(|\cS|
\cdot d_1, |\cS| \cdot d_2)=c(|\cS| \cdot d_1, |\cS| \cdot d_2)$
and for each $i$, $0 < i < |\cS|$ we have $(i \cdot d_1 ,i \cdot
d_2) \neq c(i \cdot d_1 ,i \cdot d_2)$.

Assume first that $\text{g.c.d.}(\frac{d_1 v_{22}-d_2
v_{21}}{\tau},\frac{d_2 v_{11}-d_1 v_{12}}{\tau})=1$ and
$\text{g.c.d.}(\tau , | \cS |)=1$. Assume for the contrary, that
there exist three integers $i$, $\beta_1$, and $\beta_2$, such
that $\beta_1 (v_{11},v_{12}) + \beta_2 (v_{21},v_{22}) =(i \cdot
d_1 ,i \cdot d_2 )$, $0 < i < |\cS|$. Hence we have,
$\frac{\beta_2}{\beta_1} = \frac{d_2 v_{11}-d_1 v_{12}}{d_1
v_{22}-d_2 v_{21}}=\frac{\alpha_2}{\alpha_1}$. Since
$\text{g.c.d.}(\frac{d_1 v_{22}-d_2 v_{21}}{\tau},\frac{d_2
v_{11}-d_1 v_{12}}{\tau})=1$ it follows that $\beta_1 = \gamma
\frac{d_1 v_{22}-d_2 v_{21}}{\tau}$ and $\beta_2 = \gamma
\frac{d_2 v_{11}-d_1 v_{12}}{\tau}$, for some $0 < \gamma < \tau$.
Therefore, we have $i \cdot d_1 = \beta_1 v_{11} + \beta_2 v_{21}
= \frac{ \gamma d_1 |\cS|}{\tau}$, i.e., $i = \frac{\gamma
|\cS|}{\tau}$. But, since $\text{g.c.d.}(\tau , | \cS |)=1$ it
follows that $\gamma = \rho \tau$, for some integer $\rho > 0$, a
contradiction to the fact that $0 < \gamma < \tau$. Hence, our
assumption on the existence of three integers $i$, $\beta_1$, and
$\beta_2$ is false. Thus, by Lemma~\ref{lem:tile_fold2} we have
that if $\text{g.c.d.}(\frac{d_1 v_{22}-d_2
v_{21}}{\tau},\frac{d_2 v_{11}-d_1 v_{12}}{\tau})=1$ and
$\text{g.c.d.}(\tau , | \cS |)=1$ then $(\Lambda,\cS,\delta)$
defines a folding with the direction vector $\delta =(+d_1,+d_2)$.

Assume now that $(\Lambda,\cS,\delta)$ defines a folding with the
direction vector $\delta =(+d_1,+d_2)$. Assume for the contrary
that $\text{g.c.d.}(\frac{d_1 v_{22}-d_2 v_{21}}{\tau},\frac{d_2
v_{11}-d_1 v_{12}}{\tau})=\nu_1 > 1$ or $\text{g.c.d.}(\tau , |
\cS |)= \nu_2 >1$. We distinguish now between two cases.

\noindent {\bf case 1:} If $\text{g.c.d.}(\frac{d_1 v_{22}-d_2
v_{21}}{\tau},\frac{d_2 v_{11}-d_1 v_{12}}{\tau})=\nu_1 > 1$ then
$\beta_1 = \frac{d_1 v_{22}-d_2 v_{21}}{\tau \nu_1}$ and $\beta_2
= \frac{d_2 v_{11}-d_1 v_{12}}{\tau \nu_1}$ are integers.
Therefore, $\beta_1 (v_{11},v_{12}) + \beta_2 (v_{21},v_{22})
=(\frac{|\cS| \cdot d_1}{\tau \nu_1},\frac{|\cS| \cdot d_2}{\tau
\nu_1})$. Hence, $\frac{|\cS| }{\nu_1}$ is an integer and for the
integers $\beta'_1 = \frac{d_1 v_{22}-d_2 v_{21}}{\nu_1}$ and
$\beta'_2 = \frac{d_2 v_{11}-d_1 v_{12}}{\nu_1}$ we have $\beta'_1
(v_{11},v_{12}) + \beta'_2 (v_{21},v_{22}) =(\frac{|\cS| }{\nu_1}
d_1,\frac{|\cS|}{\nu_1} d_2)$ and as a consequence by
Lemma~\ref{lem:tile_fold2} we have that $(\Lambda,\cS,\delta)$
does not define a folding, a contradiction.

\noindent {\bf case 2:} If $\text{g.c.d.}(\tau , | \cS |)= \nu_2
>1$ then let $\beta_1 = \frac{d_1 v_{22}-d_2 v_{21}}{\nu_2}$ and
$\beta_2 = \frac{d_2 v_{11}-d_1 v_{12}}{\nu_2}$. Hence, $\beta_1
(v_{11},v_{12}) + \beta_2 (v_{21},v_{22}) =(\frac{|\cS| }{\nu_2}
d_1,\frac{|\cS|}{\nu_2} d_2)$. Clearly, $\beta_1$, $\beta_2$, and
$\frac{|\cS|}{\nu_2}$ are integers, and as a consequence by
Lemma~\ref{lem:tile_fold2} we have that $(\Lambda,\cS,\delta)$
does not define a folding, a contradiction.

Therefore, if $(\Lambda,\cS,\delta)$ defines a folding with the
ternary vector $\delta =(+1,+1)$ then
$\text{g.c.d.}(v_{22}-v_{21},v_{11}-v_{12})=1$.
\end{proof}
The generalization of Theorem~\ref{thm:new_cond_fold2D} for the
$D$-dimensional case is Theorem~\ref{thm:new_general_condition}
given in Appendix A.

The next lemma is an immediate consequence from the definitions on
equivalent directions and folded-row.
\begin{lemma}
\label{lem:eq_direct} If the directions $(d_1,d_2,\ldots,d_D)$ and
$(d'_1,d'_2,\ldots,d'_D)$ are equivalent then there exists a
lattice point $(i_1,i_2,\ldots,i_D)$ such that either
$(d'_1,d'_2,\ldots,d'_D)=(i_1+d_1,i_2+d_2,\ldots,i_d+d_D)$ or
$(d'_1,d'_2,\ldots,d'_D)=(i_1-d_1,i_2-d_2,\ldots,i_d-d_D)$.
\end{lemma}

\begin{lemma}
Let $\Lambda$ be a lattice tiling for a shape $\cS$. If $|\cS|$ is
a prime number then there exist $\frac{3^D-1}{2}$ ternary
direction vectors which form folding if and only if there does not
exist a lattice point $(i_1 , i_2 , \ldots , i_D )$, where for
each $i$, $1 \leq j \leq D$, we have $|i_j| \leq 2$.
\end{lemma}
\begin{proof}
By Lemma~\ref{lem:size_fold}, if $|\cS|$ is a prime number, then
the number of elements in a folded-row for a given ternary vector
$\delta$ is either one or $|\cS|$. By
Corollary~\ref{cor:prime_fold} the number of elements is one if
and only if $\delta$ is a lattice point.

If there exist two equivalent directions $(d_1 , \ldots , d_D)$
and $(d'_1 , \ldots , d'_D)$ then by Lemma~\ref{lem:eq_direct} we
have that $(d_1 - d'_1 , \ldots , d _D - d'_D)$ is a lattice
point, where $|d_i - d'_i| \leq 2$ for each $i$, $1 \leq i \leq D$
(since $|d_i| \leq 1$ and $|d'_i| \leq 1$).

If there exists a lattice point $(i_1 , \ldots , i_D)$ for which
$|i_j| \leq 2$, $1 \leq j \leq D$, then there exists two ternary
vectors $(d_1 , \ldots , d_D)$ and $(d'_1 , \ldots , d'_D)$ for
which $(i_1 , \ldots , i_D)=(d_1 - d'_1 , \ldots , d _D - d'_D)$.
\end{proof}

The same result is obtained when $|\cS|$ is not a prime number if
the necessary conditions of
Theorem~\ref{thm:new_general_condition} are satisfied for all the
related $\frac{3^D-1}{2}$ ternary direction vectors. In any case,
if there exist a lattice point $(i_1 , i_2 , \ldots , i_D )$,
where for each $j$, $1 \leq j \leq D$, we have $|i_j| \leq 2$,
then there are some related ternary direction vectors which form
equivalent folding. We can also give an answer to this question by
finding one ternary direction vector which defines a folding and
using Corollary~\ref{cor:num_direct}.

\section{Bounds on Synchronization Patterns}
\label{sec:bounds}

Our original motivation for the generalization of the folding
operation came from the design of two-dimensional synchronization
patterns. Given a grid (square or hexagonal) and a shape $\cS$ on
the grid, we would like to find what is the largest set $\Delta$
of dots on grid points, $|\Delta|=m$, located in $\cS$, such that
the following property hold. All the $\binom{m}{2}$ lines between
dots in $\Delta$ are distinct either in their length or in their
slope. Such a shape $\cS$ with dots is called a {\it distinct
difference configuration} (DDC). If $\cS$ is an $m \times m$ array
with exactly one dot in each row and each column than $\cS$ is
called a Costas array~\cite{GoTa82}. If $\cS$ is a $k \times m$
array with exactly one dot in each column then $\cS$ is called a
sonar sequence~\cite{GoTa82}. If $\cS$ is a $k \times n$ DDC array
then $\cS$ is called a Golomb rectangle~\cite{Rob85}. These
patterns have various applications as described in~\cite{GoTa82}.
A new application of these patterns to the design of key
predistribution scheme for wireless sensor networks was described
lately in~\cite{BEMP}. In this application the shape $\cS$ might
be a Lee sphere, an hexagon, or a circle, and sometimes another
regular polygon. This application requires in some cases to
consider these shapes in the hexagonal grid. F3 was used for this
application in~\cite{BEMP08a} to form a DDC whose shape is a
rectangle rotated in 45 degrees in the square grid (see
Figure~\ref{fig:fold_diag}). Henceforth, we assume that our grid
is $\Z^D$, i.e., the square grid for $D=2$. Since the all the
results of the previous sections hold for $D$-dimensional shapes
we will continue to state the results in a $D$-dimensional
language, even so the applied part for synchronization patterns is
two-dimensional.

We will generalize some of the definition given for DDCs in
two-dimensional arrays~\cite{BEMP08a} for multidimensional arrays.
The reason is not just the generalization, but we also need these
definitions in the sequel. Let $\cA$ be a (generally infinite)
$D$-dimensional array of dots in $\Z^D$, and let $\eta_1 , \eta_2,
\ldots , \eta_D$ be positive integers. We say that $\cA$ is a
\emph{multi periodic} (or {\it doubly periodic} if $D=2$) with
period $(\eta_1 , \eta_2, \ldots , \eta_D)$ if
$\cA(i_1,i_2,\ldots,i_D)=\cA(i_1
+\eta_1,i_2,\ldots,i_D)=\cA(i_1,i_2+\eta_2,\ldots,i_D)= \cdots
=\cA(i_1,i_2,\ldots,i_D +\eta_D)$. We define the \emph{density} of
$\cA$ to be $d/(\Pi_{j=1}^D \eta_j)$, where $d$ is the number of
dots in any $\eta_1 \times \eta_2 \times \cdots \times \eta_D$
sub-array of $\cA$. Note that the period $(\eta_1 , \eta_2, \ldots
, \eta_D)$ might not be unique, but that the density of $\cA$ does
not depend on the period we choose. We say that a multi periodic
array $\cA$ of dots is a \emph{multi periodic $n_1 \times n_2
\times \cdots n_D$ DDC} if every $n_1 \times n_2 \times \cdots
n_D$ sub-array of $\cA$ is a DDC.

We write $(i_1,i_2,\ldots,i_D)+\cS$ for the shifted copy
$\{(i_1+i'_1,i_2+i'_2,\ldots,i_D+i'_D):(i'_1,i'_2,\ldots,i'_D)\in\cS\}$
of $\cS$. We say that a multi periodic array $\cA$ is a
\emph{multi periodic $\cS$-DDC} if the dots contained in every
shift $(i_1,i_2,\ldots,i_D)+\cS$ of $\cS$ form a DDC.

The definition of the density is given based on periodicity of a
$D$-dimensional box. If $\mu$ is the density, of the multi
periodic array $\cA$, it implies that given a shape $\cS$, the
average number of dots in any shape $\cS$ shifted all over $\cA$
is $\mu | \cS |$. This leads to the following theorem given
in~\cite{BEMP08a} for the two-dimensional case and which has a
similar proof for the multidimensional case.

\begin{theorem}
\label{thm:general_construction} Let $\cS$ be a shape, and let
$\cA$ be a multi periodic $\cS$-DDC of density $\mu$. Then there
exists a set of at least
$\left\lceil\mu\lvert\cS\rvert\right\rceil$ dots contained in
$\cS$ that form a DDC.
\end{theorem}

\vspace{0.2cm}

Another important observation from the definition of multi
periodic $\cS$-DDC is the following lemma from~\cite{BEMP08a}.
\begin{lemma}
\label{lem:sub_DDC} Let $\cA$ be a multi periodic $\cS$-DDC, and
let $\cS' \subseteq \cS$. Then $\cA$ is a multi periodic
$\cS'$-DDC.
\end{lemma}

\vspace{0.2cm}

Let $\cS_1$, $\cS_2, \ldots$ be an infinite sequence of similar
shapes such that $| \cS_{i+1} | > | \cS_i |$. Using the technique
of Erd\"{o}s and Tur\'{a}n~\cite{EGRT92,Erd41}, for which a
detailed proof is given in~\cite{BEMP08a}, one can prove that

\begin{theorem}
\label{thm:upper_bound} An upper bound on the number of dots in
$\cS_i$, $i \rightarrow \infty$, is $\text{lim}_{i \rightarrow
\infty} ( \sqrt{| \cS_i |} + o(\sqrt{| \cS_i |}) )$.
\end{theorem}

\vspace{0.3cm}

Let $\cS$ and $\cS'$ be two-dimensional shapes in the grid. We
will denote by $\Delta (\cS,\cS')$ the largest intersection
between $\cS$ and $\cS'$ in any orientation. Our bounds on the
number of dots in a DDC with a given shape are based on the
following result.

\begin{theorem}
\label{thm:infinite} Assume we are given a multi periodic
$\cS$-DDC array $\cA$ with density $\mu$. Let $\cQ$ be another
shape on $\Z^D$. Then there exists a copy of $\cQ$ on $\Z^D$ with
at least $\lceil \mu \cdot \Delta (\cS,\cQ) \rceil$ dots.
\end{theorem}
\begin{proof}
Let $\cQ'$ be the shape such that $\cQ' = \cS \cap \cQ$ and
$|\cQ'| = \Delta (\cS,\cQ)$. By Lemma~\ref{lem:sub_DDC} we have
that $\cA$ is a multi periodic $\cQ'$-DDC. By
Theorem~\ref{thm:general_construction}, there exists a set of at
least $\left\lceil\mu\lvert\cQ'\rvert\right\rceil$ dots contained
in $\cS$ that form a DDC. Thus, there exists a copy of $\cQ$ on
$\Z^D$ with at least $\lceil \mu \cdot \Delta (\cS,\cQ) \rceil$
dots.
\end{proof}

In order to apply Theorem~\ref{thm:infinite} we will use folding
of sequences defined as follows. Let $A$ be an abelian group, and
let $\cB=\{b_1,b_2,\ldots,b_m\}\subseteq A$ be a sequence of $m$
distinct elements of $A$. We say that $\cB$ is a
\emph{$B_2$-sequence over $A$} if all the sums $a_{i_1}+a_{i_2}$
with $1\leq i_1\leq i_2\leq m$ are distinct. For a survey on
$B_2$-sequences and their generalizations the reader is referred
to~\cite{Bry04}. The following lemma is well known and can be
readily verified.

\begin{lemma}
\label{lem:B2_diff} A subset $\cB=\{a_1,a_2,\ldots,a_m\}\subseteq A$
is a $B_2$-sequence over $A$ if and only if all the differences
$a_{i_1}-a_{i_2}$ with $1\leq i_1 \neq i_2\leq m$ are distinct in
$A$.
\end{lemma}

Note that if $\cB$ is a $B_2$-sequence over $\Z_n$ and $a\in\Z_n$,
then so is the shift $a+B=\{a+e:e\in B\}$. The following theorem,
due to Bose~\cite{Bose42}, shows that large $B_2$-sequences over
$\Z_n$ exist for many values of $n$.

\begin{theorem}
\label{thm:Bose} Let $q$ be a prime power. Then there exists a
$B_2$-sequence $a_1,a_2,\ldots,a_m$ over $\Z_n$ where $n=q^2-1$
and $m=q$.
\end{theorem}

\subsection{A Lattice Coloring for a Given Shape}

In this subsection we will describe how we apply folding to obtain
a DDC with a shape $\cS$ and a multi periodic $\cS$-DDC. Let
$\Lambda$ be a lattice tiling for $\cS$ and let $\delta = ( d_1 ,
d_2 , \ldots , d_D )$ be a direction vector such that
$(\Lambda,\cS,\delta)$ defines a folding. We assign an integer
from $\Z_n$, $n= | \cS |$, to each point of $\Z^D$. The {\it
lattice coloring} $\cC( \Lambda , \delta )$ is defined as follows.
We assign 0 to the point $(0,0, \ldots ,0)$ and we color the next
element on the folded-row with 1 and so on until $|\cS|-1$ to the
last element on the folded-row. This complete the coloring of the
points of the shape $\cS$ whose center is the origin. To position
$(i_1,i_2, \ldots , i_D )$ we assign the color of position
$(i_1,i_2, \ldots , i_D )-c(i_1,i_2, \ldots , i_D )$. The color of
position $(i_1,i_2,\ldots,i_D)$ will be denoted by
$\cC(i_1,i_2,\ldots,i_D)$.

We will generalize the definition of folding a sequence into a
shape $\cS$ by the direction $\delta$, given the lattice tiling
$\Lambda$ for $\cS$. The folding of a sequence $\cB=b_0 b_1 ~
\ldots ~ b_{n-1}$ into an array colored by the elements of $\Z_n$
is defined by assigning the value $b_i$ to all the points of the
array colored with the color $i$. If the coloring was defined by
the use of the folding as described in this subsection, we say
that the array is defined by $(\Lambda,\cS,\delta,\cB)$. Note,
that we use the same notation for folding the sequence $\cB$ into
the shape $\cS$. The one to which we refer should be understood
from the context.

Given a point $(i_1,i_2,\ldots,i_D) \in \Z^D$, we say that the set
of points $\{ ( i_1 + \ell \cdot d_1 , i_2 + \ell \cdot d_2 ,
\ldots, i_D + \ell \cdot d_D ) ~:~ \ell \in \Z \}$ is a {\it row
of $\Z^D$ defined by $\delta$}. This is also the row of
$(i_1,i_2,\ldots,i_D)$ defined by $\delta= (d_1 , d_2 , \ldots ,
d_D )$.

\begin{lemma}
\label{lem:diag_lattice} If the triple $(\Lambda,\cS,\delta)$
defines a folding then in any row of $\Z^D$ defined by $\delta$
there are lattice points.
\end{lemma}
\begin{proof}
Given a point $(i_1,i_2,\ldots,i_D)$ and its color
$\cC(i_1,i_2,\ldots,i_D)$, then by the definitions of the folding
and the coloring we have that
$\cC(i_1+d_1,i_2+d_2,\ldots,i_D+d_D)\equiv
\cC(i_1,i_2,\ldots,i_D)+1 ~(\bmod ~ |\cS|$). Hence, the row
defined by $\delta$ has all the values between 0 and $|\cS|-1$ in
their natural order modulo $|\cS|$. Therefore, any row defined by
$\delta$ has lattice points (which are exactly the points of 
this row which are colored with {\it zeroes}).
\end{proof}

\begin{cor}
\label{cor:eq_diff_col} If  $(i_1,i_2,\ldots,i_D)$,
$(i_1+e_1,i_2+e_2,\ldots,i_D+e_D)$, $(j_1,j_2,\ldots,j_D)$, and
$(j_1+e_1,j_2+e_2,\ldots,j_D+e_D)$ are four points of $\Z^D$ then
$\cC(i_1+e_1,i_2+e_2,\ldots,i_D+e_D)-\cC(i_1,i_2,\ldots,i_D)
\equiv
\cC(j_1+e_1,j_2+e_2,\ldots,j_D+e_D)-\cC(j_1,j_2,\ldots,j_D)~(\bmod
~ |\cS|)$.
\end{cor}
\begin{proof}
By Lemma~\ref{lem:diag_lattice} to each one of these four points
there exists a lattice point in its row defined by $\delta$. Let
\begin{itemize}
\item $P_1 = (i_1 + \alpha_1 \cdot d_1 ,i_2 + \alpha_1 \cdot d_2 ,
\ldots, i_D +  \alpha_1 \cdot d_D )$ be the lattice point in the
row of $(i_1,i_2,\ldots,i_D)$;

\item $P_2 = (j_1 + \alpha_2 \cdot d_1,j_2 + \alpha_2 \cdot d_2 ,
\ldots, j_D + \alpha_2 \cdot d_D )$ the lattice point in the row
of $(j_1,j_2,\ldots,j_D)$;

\item $P_3 = ((i_1 + e_1 ) + \alpha_3 \cdot d_1 ,(i_2 +e_2) +
\alpha_3 \cdot d_2 , \ldots, (i_D+e_D) + \alpha_3 \cdot  d_D )$
the lattice point in the row of
$(i_1+e_1,i_2+e_2,\ldots,i_D+e_D)$.
\end{itemize}
Therefore, $P_4 = P_2 + P_3 - P_1 =((j_1 + e_1
)+(\alpha_2+\alpha_3-\alpha_1) \cdot d_1 , (j_2 +e_2)
+(\alpha_2+\alpha_3-\alpha_1) \cdot d_2 , \ldots, (j_D+e_D) +
(\alpha_2+\alpha_3-\alpha_1) \cdot d_D )$ is also a lattice point.
$P_4$ is a lattice point in the row, defined by $\delta$, of
$(j_1+e_1,j_2+e_2,\ldots,j_D+e_D)$. All these four points are
colored with {\it zeroes}. Hence, $\cC (i_1,i_2,\ldots,i_D) \equiv
- \alpha_1~(\bmod~|\cS|)$, $\cC (i_1+e_1,i_2+e_2,\ldots,i_D+e_D)
\equiv - \alpha_3~(\bmod~|\cS|)$, $\cC (j_1,j_2,\ldots,j_D) \equiv
- \alpha_2~(\bmod~|\cS|)$, and $\cC
(j_1+e_1,j_2+e_2,\ldots,j_D+e_D) \equiv -
(\alpha_2+\alpha_3-\alpha_1)~(\bmod~|\cS|)$. Now, the claim of the
corollary is readily verified.
\end{proof}
\begin{cor}
If $\delta'$ is an integer vector of length $D$ then there exists
an integer $e(\delta')$ such that for any given point
$P=(i_1,i_2,\ldots,i_D)$ we have $\cC(P + \delta') =
\cC(P)+e(\delta')~(\bmod~|\cS|)$.
\end{cor}
\begin{cor}
\label{cor:fold_periodic} If the triple $(\Lambda,\cS,\delta)$
defines a folding and $\cB$ is a $B_2$-sequence over $\Z_n$, where
$n=|\cS|$, then the array $\cA$ defined by
$(\Lambda,\cS,\delta,\cB)$ is multi periodic.
\end{cor}
\begin{proof}
Clearly, the array has period $(|\cS|,|\cS|,\ldots,|\cS|)$ and the
result follows.
\end{proof}
\begin{theorem}
\label{thm:fold_B2} If the triple $(\Lambda,\cS,\delta)$ defines a
folding and $\cB$ is a $B_2$-sequence over $\Z_n$, where
$n=|\cS|$, then the pattern of dots defined by
$(\Lambda,\cS,\delta,\cB)$ is a multi periodic $\cS$-DDC.
\end{theorem}
\begin{proof}
By Corollary~\ref{cor:fold_periodic} the constructed array is
multi periodic.

Since $(\Lambda,\cS,\delta)$ defines a folding it follows that the
$|\cS|$ colors inside the shape $\cS$ centered at the origin are
all distinct. By Corollary~\ref{cor:eq_diff_col}, for the four
positions $(i_1,i_2,\ldots,i_D)$,
$(i_1+e_1,i_2+e_2,\ldots,i_D+e_D)$, $(j_1,j_2,\ldots,j_D)$, and
$(j_1+e_1,j_2+e_2,\ldots,j_D+e_D)$ we have that
$\cC(i_1+e_1,i_2+e_2,\ldots,i_D+e_D)-\cC(i_1,i_2,\ldots,i_D)
\equiv
\cC(j_1+e_1,j_2+e_2,\ldots,j_D+e_D)-\cC(j_1,j_2,\ldots,j_D)~(\bmod~
|\cS|)$. Hence, at most three of these integers (colors) are
contained in $\cB$. It implies that if these four points belong to
the same copy of $\cS$ on the grid then at most three of these
points have dots, since the dots are distributed by the
$B_2$-sequence $\cB$. Thus, any shape $\cS$ on $\Z^D$ will define
a DDC and the theorem follows.
\end{proof}

\begin{cor}
\label{cor:fold_B2} If the triple $(\Lambda,\cS,\delta)$ defines a
folding and $\cB$ is a $B_2$-sequence over $\Z_n$, where
$n=|\cS|$, then the pattern of dots defined by
$(\Lambda,\cS,\delta,\cB)$ is a DDC.
\end{cor}
Note, that the difference between Theorem~\ref{thm:fold_B2} and
Corollary~\ref{cor:fold_B2} is related to the folding into $\Z^D$
and $\cS$, respectively. The last lemma is given for completeness.

\begin{lemma}
\label{lem:disjoint_colors} If $(\Lambda,\cS,\delta)$ defines a
folding then the $|\cS|$ colors inside any copy of $\cS$ on a
$\Z^D$ are all distinct.
\end{lemma}
\begin{proof}
Let $\cS_1$ and $\cS_2$ be two distinct copies of $\cS$ on $\Z^D$.
Clearly, $\cS_2 = (e_1 , \ldots , e_D ) + \cS_1$. By
Corollary~\ref{cor:eq_diff_col}, for each $(i_1 , \ldots , i_D
),(j_1 , \ldots , j_D ), \in \cS_1$ we have $\cC(i_1+e_1 , \ldots
, i_D + e_D) - \cC(i_1 , \ldots , i_D ) \equiv \cC(j_1+e_1 ,
\ldots , j_D + e_D) - \cC(j_1 , \ldots , j_D ) ~(\bmod~|\cS|)$.
Therefore, if $\cS_1$ contains $|\cS|$ distinct colors then also
$\cS_2$ contains $|\cS|$ distinct colors. The lemma follows now
from the fact that $(\Lambda,\cS,\delta)$ defines a folding and
therefore all the colors in the shape $\cS$ whose center is in the
origin are distinct.
\end{proof}
Note, that theorem~\ref{thm:fold_B2} is also an immediate
consequence of Lemma~\ref{lem:disjoint_colors}.

\section{Bounds for Specific Shapes}
\label{sec:DDAs}

In this section we will present some lower bounds on the number of
dots in some two-dimensional DDCs with specific shapes. In the
sequel we will use Theorem~\ref{thm:infinite},
Theorem~\ref{thm:fold_B2}, and Corollary~\ref{cor:fold_B2} to form
DDCs with various given shapes with a large number of dots. To
examine how good are our lower bounds on the number of dots, in a
DDC whose shape is $\cQ$, we should know what is the upper bound
on the number of dots in a DDC whose shape is $\cQ$. By
Theorem~\ref{thm:upper_bound} we have that for a DDC whose shape
is a regular polygon or a circle, an upper bound on the number of
dots is at most $\sqrt{s} +o(\sqrt{s})$, where the shape contains
$s$ points of the square grid and $s \rightarrow \infty$. One of
the main keys of our constructions, and the usage of the given
theory, is the ability to produce a multi periodic $\cS$-DDC,
where $\cS$ is a rectangle, the ratio between its sides is close
to any given number $\gamma$, and if its area is $s$ then the
number of dots in it is approximately $\sqrt{s} +o(\sqrt{s})$.

For the construction we will need the well known Dirichlet's
Theorem~\cite[p. 27]{NiZu79}.

\begin{theorem}
\label{thm:Dirichlet} If $a$ and $b$ are two positive relatively
primes integers then the arithmetic progression of terms $ai+b$,
for $i=1, 2, ...$, contains an infinite number of primes.
\end{theorem}

The following theorem is a well known consequence of the well
known Euclidian algorithm~\cite[p. 11]{NiZu79}.

\begin{theorem}
\label{thm:Euclid} If $\alpha$ and $\beta$ are two integers such
that $\text{g.c.d.}(\alpha,\beta)=1$ then there exist two integers
$c_\alpha$ and $c_\beta$ such that $c_\alpha \alpha + c_\beta
\beta =1$.
\end{theorem}

The next theorem makes usage of these well known old foundations.

\begin{theorem}
\label{thm:ratio_rec} For each positive number $\gamma$ and any
$\epsilon > 0$, there exist two integers $n_1$ and $n_2$ such that
$\gamma \leq \frac{n_1}{n_2} < \gamma + \epsilon$; and there
exists a multi periodic $\cS$-DDC with $\sqrt{a \cdot b} R + o(R)$
dots whose shape is an $n_1 \times n_2 = (a R +o(R)) \times (b R
+o(R))$ rectangle, where $n_1 n_2 =p^2-1$ for some prime $p$, and
$n_1$ is an even integer.
\end{theorem}
\begin{proof}
Given a positive number $\gamma$ and an $\epsilon > 0$, it is easy
to verify that there exist two integers $\alpha$ and $\beta$ such that
$\sqrt{\gamma} \leq \frac{\beta}{\alpha} < \sqrt{\gamma + \epsilon}$ and
$\text{g.c.d.}(\alpha,\beta)=2$. By Theorem~\ref{thm:Euclid} there
exist two integers $c_\alpha$, $c_\beta$ such that either
$c_\alpha \alpha +2 = c_\beta \beta >0$ or $c_\beta \beta +2 =
c_\alpha \alpha >0$.

Assume $c_\alpha \alpha +2 = c_\beta \beta
>0$ (the case where $c_\beta \beta +2 =
c_\alpha \alpha >0$ is handled similarly). Clearly, any factor of
$\alpha$ cannot divide $c_\alpha \alpha +1$. Since $\beta$ divides
$c_\alpha \alpha +2$, it follows that a factor of $\beta$ cannot
divide $c_\alpha \alpha +1$. Hence,
$\text{g.c.d.}(\alpha\beta,c_\alpha \alpha +1)=1$. Therefore, by
Theorem~\ref{thm:Dirichlet} there exist infinitely many primes in
the sequence $\alpha \beta R +c_\alpha \alpha +1$,
$R=1,2,\ldots$~.

Let $p$ be a prime number of the form $\alpha \beta R +c_\alpha
\alpha +1$. Now, $p^2-1=(p+1)(p-1)=(\alpha \beta R +c_\alpha
\alpha +2)(\alpha \beta R +c_\alpha \alpha)=(\alpha \beta R
+c_\beta \beta)(\alpha \beta R +c_\alpha \alpha)=(\alpha^2 R +
\alpha c_\beta)(\beta^2 R + \beta c_\alpha)$. Thus, a $(\beta^2 R
+ \beta c_\alpha) \times (\alpha^2 R + \alpha c_\beta)$ rectangle
satisfies the size requirements for the $n_1 \times n_2$ rectangle
of the Theorem.

Let $a=\beta^2$,
$b=\alpha^2$, $n_1 = \beta^2 R + \beta c_\alpha$,
$n_2 = \alpha^2 R + \alpha c_\beta$, and let $\cS$ be an $n_1 \times n_2$ rectangle.
Let $\Lambda$ be the a lattice
tiling for $\cS$ with the generator matrix
$$
G=\left[\begin{array}{cc}
n_2 & \frac{n_1}{2}+\theta \\
0 & n_1
\end{array}\right]~,
$$
where $\theta=1$ if $n_1 \equiv 0~(\bmod 4)$ and $\theta=2$ if $n_1
\equiv 2~(\bmod 4)$. By Theorem~\ref{thm:cond_fold2D},
$(\Lambda,\cS,\delta)$, $\delta=(+1,0)$, defines a folding.

The existence of a multi periodic $\cS$-DDC with $\sqrt{a
\cdot b} R + o(R)$ dots follows now
from Theorems~\ref{thm:Bose} and~\ref{thm:fold_B2}.
\end{proof}

The next key structure in our constructions is a certain family of
hexagons defined next. A {\it centroid hexagon} is an hexagon with
three disjoint pairs of parallel sides. If the four angles of two
parallel sides (called the {\it bases} of the hexagon) are equal
and the four other sides are equal, the hexagon will be called a
{\it quasi-regular hexagon} and will be denoted by QRH($w,b,h$),
where $b$ is the length of a base, $h$ is the distance between the
two bases, and $b+2w$ is the length between the two vertices not
on the bases. We will call the line which connects these two
vertices, the {\it diameter} of the hexagon (even if it might not
be the longest line between two points of the hexagon).
Quasi-regular hexagon will be the shape $\cS$ that will have the
role of $\cS$ when we will apply Theorem~\ref{thm:infinite} to
obtain a lower bound on the number of dots in a shape $\cQ$ which
usually will be a regular polygon. In the sequel we will say that
$\frac{\beta}{\alpha} \approx \gamma$, when we means that $\gamma
\leq \frac{\beta}{\alpha} < \gamma + \epsilon$.

We want to show that there exists a quasi-regular hexagon
QRH($w,b,h$) with approximately $\sqrt{(b+w)h} + o(\sqrt{(b+w)h})$
dots. By Theorem~\ref{thm:ratio_rec}, there exists a doubly
periodic $\cS$-DCC, where $\cS$ is an $n_1 \times n_2 = (\alpha R
+o(R)) \times (\beta R+o(R))$ rectangle, such that
$\frac{n_2}{n_1} \approx \frac{b+w}{h}$, $n_1 n_2 = p^2 -1$ for
some prime $p$, and $n_1$ is an even integer. The lattice
$\Lambda$ of Theorem~\ref{thm:ratio_rec} is also a lattice tiling
for a a shape $\cS'$, where $\cS'$ is "almost" a a quasi-regular
hexagon QRH($w,b,h$) (part of this lattice tiling is depicted in
Figure~\ref{fig:rec_hex}). By Theorem~\ref{thm:cond_fold2D},
$(\Lambda,\cS,\delta)$, $\delta=(+1,0)$, defines a folding for
this shape too. Hence, we obtain a doubly periodic $\cS'$-DCC,
where $\cS'$ is "almost" a a quasi-regular hexagon QRH($w,b,h$)
with approximately $\sqrt{(b+w)h} + o(\sqrt{(b+w)h})$ dots. This
construction implies the following theorem.

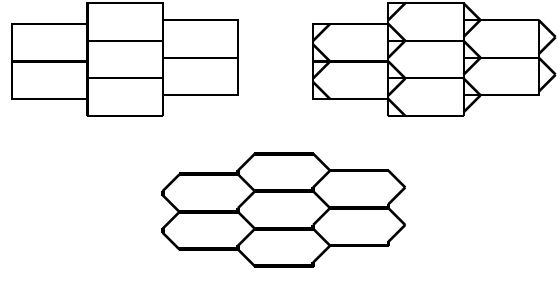
\begin{figure}[tb]

\vspace{2cm}


\setlength{\unitlength}{.5mm}
\begin{picture}(40,40)(40,0)
\linethickness{.4 pt}

\put(50,50){\line(1,0){20}} \put(50,50){\line(0,1){20}}
\put(50,70){\line(1,0){20}} \put(50,60){\line(1,0){20}}
\put(70,45.5){\line(0,1){30}} \put(90,45.5){\line(0,1){30}}
\put(70,45.5){\line(1,0){20}} \put(70,55.5){\line(1,0){20}}
\put(70,65.5){\line(1,0){20}} \put(70,75.5){\line(1,0){20}}
\put(110,51){\line(0,1){20}} \put(90,51){\line(1,0){20}}
\put(90,61){\line(1,0){20}} \put(90,71){\line(1,0){20}}

\put(130,50){\line(1,0){20}} \put(130,50){\line(0,1){20}}
\put(130,70){\line(1,0){20}} \put(130,60){\line(1,0){20}}
\put(150,45.5){\line(0,1){30}} \put(170,45.5){\line(0,1){30}}
\put(150,45.5){\line(1,0){20}} \put(150,55.5){\line(1,0){20}}
\put(150,65.5){\line(1,0){20}} \put(150,75.5){\line(1,0){20}}
\put(190,51){\line(0,1){20}} \put(170,51){\line(1,0){20}}
\put(170,61){\line(1,0){20}} \put(170,71){\line(1,0){20}}

\linethickness{1.0 pt}

\put(130,55.5){\line(1,1){4.5}} \put(130,54.5){\line(1,-1){4.5}}
\put(130,65.5){\line(1,1){4.5}} \put(130,64.5){\line(1,-1){4.5}}
\put(150,51){\line(1,1){4.5}} \put(150,50){\line(1,-1){4.5}}
\put(150,61){\line(1,1){4.5}} \put(150,60){\line(1,-1){4.5}}
\put(150,71){\line(1,1){4.5}} \put(150,70){\line(1,-1){4.5}}
\put(170,56.5){\line(1,1){4.5}} \put(170,55.5){\line(1,-1){4.5}}
\put(170,66.5){\line(1,1){4.5}} \put(170,65.5){\line(1,-1){4.5}}
\put(170,46.5){\line(1,1){4.5}} \put(170,75.5){\line(1,-1){4.5}}
\put(190,52){\line(1,1){4.5}} \put(190,61){\line(1,-1){4.5}}
\put(190,62){\line(1,1){4.5}} \put(190,71){\line(1,-1){4.5}}

\linethickness{1.0 pt}

\put(94.5,10){\line(1,0){15.5}} \put(90,14.5){\line(0,1){1}}
\put(90,24.5){\line(0,1){1}} \put(94.5,30){\line(1,0){15.5}}
\put(94.5,20){\line(1,0){15.5}} \put(110,10){\line(0,1){1}}
\put(110,20){\line(0,1){1}} \put(110,30){\line(0,1){1}}
\put(130,5.5){\line(0,1){1}} \put(130,15.5){\line(0,1){1}}
\put(130,25.5){\line(0,1){1}} \put(114.5,5.5){\line(1,0){15.5}}
\put(114.5,15.5){\line(1,0){15.5}}
\put(114.5,25.5){\line(1,0){15.5}}
\put(114.5,35.5){\line(1,0){15.5}} \put(150,11){\line(0,1){1}}
\put(150,21){\line(0,1){1}} \put(134.5,11){\line(1,0){15.5}}
\put(134.5,21){\line(1,0){15.5}} \put(134.5,31){\line(1,0){15.5}}

\put(90,15.5){\line(1,1){4.5}} \put(90,14.5){\line(1,-1){4.5}}
\put(90,25.5){\line(1,1){4.5}} \put(90,24.5){\line(1,-1){4.5}}
\put(110,11){\line(1,1){4.5}} \put(110,10){\line(1,-1){4.5}}
\put(110,21){\line(1,1){4.5}} \put(110,20){\line(1,-1){4.5}}
\put(110,31){\line(1,1){4.5}} \put(110,30){\line(1,-1){4.5}}
\put(130,16.5){\line(1,1){4.5}} \put(130,15.5){\line(1,-1){4.5}}
\put(130,26.5){\line(1,1){4.5}} \put(130,25.5){\line(1,-1){4.5}}
\put(130,6.5){\line(1,1){4.5}} \put(130,35.5){\line(1,-1){4.5}}
\put(150,12){\line(1,1){4.5}} \put(150,21){\line(1,-1){4.5}}
\put(150,22){\line(1,1){4.5}} \put(150,31){\line(1,-1){4.5}}

%
%

\end{picture}

\caption{From rectangle to "almost" quasi-perfect hexagon with the
same lattice tiling} \label{fig:rec_hex}
\end{figure}

\begin{theorem}
\label{thm:reg_hex}
A lower bound on the number of dots
in a regular hexagon with sides of length $R$ is approximately
$\frac{\sqrt{3\sqrt{3}}}{\sqrt{2}} R+o(R)$.
\end{theorem}

Now, we can give a few examples for other specific shapes, mostly,
regular polygons. To have some comparison between the bounds for
various shapes we will assume that the radius of the circle or the
regular polygons is $R$ (the {\it radius} is the distance from the
center of the regular polygon to any one its vertices). We also
define the {\it packing ratio} as the ratio between the lower and
the upper bounds on the number of dots. The shape $\cS$ that we
use will always by a multi periodic $\cS$-DDC on a multi periodic
array $\cA$.

\subsection{Circle}
\label{subsec:CircSqu}

We apply Theorem~\ref{thm:infinite}, where $\cS$ is a regular
hexagon with radius $\rho$ and $\cQ$ is a circle with radius $R$,
sharing the same center. The upper bound on the number of dots in
$\cQ$ is $\sqrt{\pi} R + o(R)$. A lower bound on the number of
dots in $\cS$ is approximately $\frac{\sqrt{3\sqrt{3}}}{\sqrt{2}}
\rho +o(\rho)$ and hence the density of $\cA$ is approximately
$\frac{\sqrt{2}}{\sqrt{3 \sqrt{3}}\rho}$. Let $\theta$ be the
angle between two radius lines to the two intersection points of
the hexagon and the circle on one edge of the hexagon. We have
that $\Delta ( \cS , \cQ )=(\pi - 3 \theta + 3 \sin \theta) R^2$
and $\rho = \frac{\cos \frac{\theta}{2}}{\cos \frac{\pi}{6}} R$.
Thus, a lower bound on the number of dots in $\cQ$ is
$\frac{\sqrt{3\sqrt{3}} \rho+o(\rho)}{\sqrt{2}|\cS|} \Delta
(\cS,\cQ)$. The maximum is obtained when $\theta = 0.536267$
yielding a lower bound of $1.70813R + o(R)$ on the number of dots
in $\cQ$ and a packing ratio of 0.9637.


We must note again, that even so this construction works for
infinitely many values of $R$, the density of these values is
quite low. This is a consequence of Theorem~\ref{thm:ratio_rec}
which can be applied for an arbitrary ratio $\gamma$ only when the
corresponding integers obtained by Dirichlet's Theorem are primes.
Of course, there are many possible ratios between the sides of the
rectangle that can be obtained for infinitely many values. A
simple example is for any factorization of $p^2 -1 = n_1 n_2$ we
can form an $n_1 \times n_2$ DDC and from its related
quasi-regular hexagons. We won't go into details to obtain bounds
which hold asymptotically for any given $R$ as we conjecture that
the construction for quasi-regular hexagon can be strengthen
asymptotically for almost all parameters. Nevertheless, we will
show briefly how we can use a doubly periodic $\cS$-DDC, where
$\cS$ is a square to obtain a lower bound for the number of dots
in a DDC whose shape is a circle. We use a doubly periodic
$\cS$-DDC, where $\cS$ is a $(p+1) \times (p-1)$ rectangle. For a
lattice tiling of $\cS$ we use a lattice $\Lambda$ with the
generator matrix

$$
G=\left[\begin{array}{cc}
p-1 & \frac{p+1}{2}+\theta \\
0 & p+1
\end{array}\right]~,
$$
where $\theta=1$ if $p+1 \equiv 0~(\bmod~4)$ and $\theta=2$ if
$p+1 \equiv 2~(\bmod~4)$. By Theorem~\ref{thm:cond_fold2D},
$(\Lambda,\cS,\delta)$, $\delta=(+1,0)$, defines a folding. We can
use Theorem~\ref{thm:trans_shape} to obtain a new shape $\cS'$
which produces better intersection with a circle, and a better
lower bound on the number of dots in it (the previous best packing
ratio obtained with the method implied only by
Theorem~\ref{thm:infinite} (without using
Theorem~\ref{thm:trans_shape} and better multi periodic
$\cS$-DDCs) was $0.91167$ and it was given in~\cite{BEMP08a}).

\subsection{Regular Polygon}

For regular polygons with small number of sides we will use
specific constructions which are given in Appendix C. If the
number of sides is large we will use Theorem~\ref{thm:infinite},
where $\cQ$ will be the regular polygon and $\cS$ is a regular
hexagon. Assume that the regular polygon has $n$ sides, $R$ is its
radius, and $\rho$ is the radius of the regular hexagon. The area
of the hexagon is $\frac{3 \sqrt{3}}{2} \rho^2$ and hence the
density of the doubly periodic array $\cA$ is approximately
$\frac{\sqrt{2}}{\sqrt{3 \sqrt{3}}\rho}$. The area of the regular
polygon is $\frac{n \cdot R^2 \sin \frac{2 \pi}{n}}{2}$ and hence
an upper bound on the number of dots in $\cQ$ is $\frac{\sqrt{n
\cdot \sin \frac{2 \pi}{n}}}{\sqrt{2}}R +o(R)$. For simplicity we
will further assume that $n=12k$ (the results for other values of
$n$ are similar, but the constructions become slightly more
complicated for short description. We will choose a regular
hexagon which has a joint center with the regular polygon. We
further choose it in a way that $\cS$ and $\cQ$ intersect in
exactly 12 vertices of $\cQ$ equally spread. We will also make
sure that each side of $\cS$ intersects exactly two vertices of
$\cQ$ with equal distance from the nearest vertices of $\cS$ to
these two intersection points. It implies that $\Delta (\cS,\cQ) =
\frac{6+ n \cdot \sin \frac{2 \pi}{n}}{4} R^2$ and hence  a lower
bound on the number of dots is $\frac{6+n \cdot \sin \frac{2
\pi}{n}}{2 \cdot 3^{\frac{1}{4}} ( \sqrt{3}+1)} R +o(R)$. Some
values obtained from this construction are given in
Table~\ref{tab:boundsummary}.

For small values of $n$, specific constructions are given in
Appendix C. For some constructions we need DDCs with other shapes
like a Corner and a Flipped T which are defined in Appendix B,
where also constructions of multi periodic $\cS$-DDCs for these
shapes are given. Table~\ref{tab:boundsummary} summarizes the
bounds we obtained for regular polygons and a circle in the square
grid. The same techniques can be used for any $D$-dimensional
shape. Finally, we note that the problem is of interest also from
discrete geometry point of view. Some similar questions can be
found in~\cite{LeTh95}.

\begin{table}
\caption{Bounds on the number of dots in an $n$-gon
DDC}\label{tab:boundsummary}
\begin{equation*}
\begin{array}{lccc}
\hline
$n$ & \text{upper bound} & \text{lower bound} & \text{packing ratio} \\
\hline
3 & 1.13975R & 1.02462R & 0.899   \\
4 &1.41421R & 1.41421R &  1  \\
5 & 1.54196R & 1.45992R  & 0.9468  \\
6 &1.61185R & \approx 1.61185R & \approx 1  \\
7 & 1.65421R & 1.58844R & 0.960241  \\
8 &1.68179R & 1.62625R & 0.966977 \\
9 & 1.70075R & 1.63672R & 0.96235  \\
10 &1.71433R & 1.64786R & 0.961229  \\
12 &1.73205R & 1.66871R & 0.963433  \\
24 &1.76234R & 1.69815R & 0.963578  \\
36 &1.76796R & 1.70367R & 0.963636  \\
48 &1.76992R & 1.7056R & 0.963658  \\
60 &1.77083R & 1.7065R & 0.963669  \\
72 &1.77133R & 1.70699R & 0.963675  \\
84 &1.77163R & 1.70728R & 0.963679  \\
96 &1.77182R & 1.70747R & 0.963681  \\
\text{circle} &1.77245R & 1.70813R & 0.963708 \\
\hline
\end{array}
\end{equation*}
\end{table}

\section{Folding in the Hexagonal Grid}
\label{sec:hex}

The questions concerning DDCs can be asked in the hexagonal grid
in the same way that they are asked in the square grid. Similarly,
they can be asked in dense $D$-dimensional lattices. In this
section we will consider some part of our discussion related to
the hexagonal grid. The hexagonal grid is a two-dimensional grid
and hence we will compare it to $\Z^2$. In $\Z^2$ there are four
different ternary direction vectors, while in the hexagonal grid
there are three different related directions. But, the total
number of directions depend on the shape in both grids (see
Subsection~\ref{sec:further_generalize} and especially
Corollary~\ref{cor:num_direct}). We can define a folded-row and
folding in the hexagonal grid in the same way as they are defined
in $\Z^2$. To prove that the results remain unchanged we will
describe the well known transformation between the hexagonal grid
and $\Z^2$.

The {\it hexagonal grid} is defined as follows. We start by tiling
the plane $\R^2$ with regular hexagons whose sides have length
$1/\sqrt{3}$ (so that the centers of hexagons that share an edge
are at distance $1$). The center points of the hexagons are the
points of the grid. The hexagons tile $\R^2$ in a way that each
point $(i,0)$, $i \in \Z$, is a center of some hexagon.

The transformation uses an isomorphic representation of the
hexagonal grid. Each point $(x,y)\in\Z^2$ has the following
neighboring vertices,
$$\{(x+a,y+b) ~|~ a,b\in\{-1,0,1\}, a+b\neq 0\}.$$ It may be shown
that the two representations are isomorphic by using the mapping
$\xi: \R^2 \rightarrow \R^2$, which is defined by
$\xi(x,y)=(x+\frac{y}{\sqrt{3}},\frac{2y}{\sqrt{3}})$. The effect
of the mapping on the neighbor set is shown in
Fig.~\ref{fig:hexmodel}. From now on, slightly changing notation,
we will also refer to this representation as the hexagonal grid.
Using this new representation the neighbors of point $(i,j)$ are
\begin{multline*}
\{(i-1,j-1),(i-1,j),(i,j-1),(i,j+1),\hfill\\
\hfill (i+1,j),(i+1,j+1)\}.
\end{multline*}

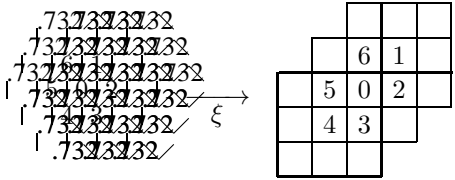
\begin{figure}[t]
\centering \setlength{\unitlength}{.4mm}
\begin{picture}(150,70)(-10,-40)
\put(0,5){\line(0,1){5.77350269}}
\put(10,5){\line(0,1){5.77350269}}
\put(20,5){\line(0,1){5.77350269}}
\put(30,5){\line(0,1){5.77350269}}
\put(10,10.77350269){\line(-1.732,1){5.0}}
\put(0,10.77350269){\line(1.732,1){5.0}}
\put(20,10.77350269){\line(-1.732,1){5.0}}
\put(10,10.77350269){\line(1.732,1){5.0}}
\put(30,10.77350269){\line(-1.732,1){5.0}}
\put(20,10.77350269){\line(1.732,1){5.0}}

\put(-5,-3.66033872){\line(0,1){5.77350269}}
\put(5,-3.66033872){\line(0,1){5.77350269}}
\put(15,-3.66033872){\line(0,1){5.77350269}}
\put(25,-3.66033872){\line(0,1){5.77350269}}
\put(35,-3.66033872){\line(0,1){5.77350269}}
\put(5,2.11316397){\line(-1.732,1){5}}
\put(-5,2.11316397){\line(1.732,1){5}}
\put(15,2.11316397){\line(-1.732,1){5}}
\put(5,2.11316397){\line(1.732,1){5}}
\put(25,2.11316397){\line(-1.732,1){5}}
\put(15,2.11316397){\line(1.732,1){5}}
\put(35,2.11316397){\line(-1.732,1){5}}
\put(25,2.11316397){\line(1.732,1){5}}

\put(-10,-12.3206774){\line(0,1){5.77350269}}
\put(0,-12.3206774){\line(0,1){5.77350269}}
\put(10,-12.3206774){\line(0,1){5.77350269}}
\put(20,-12.3206774){\line(0,1){5.77350269}}
\put(30,-12.3206774){\line(0,1){5.77350269}}
\put(40,-12.3206774){\line(0,1){5.77350269}}
\put(0,-6.54717471){\line(-1.732,1){5}}
\put(-10,-6.54717471){\line(1.732,1){5}}
\put(10,-6.54717471){\line(-1.732,1){5}}
\put(0,-6.54717471){\line(1.732,1){5}}
\put(20,-6.54717471){\line(-1.732,1){5}}
\put(10,-6.54717471){\line(1.732,1){5}}
\put(30,-6.54717471){\line(-1.732,1){5}}
\put(20,-6.54717471){\line(1.732,1){5}}
\put(40,-6.54717471){\line(-1.732,1){5}}
\put(30,-6.54717471){\line(1.732,1){5}}

\put(-5,-15.20751351){\line(0,-1){5.77350269}}
\put(5,-15.20751351){\line(0,-1){5.77350269}}
\put(15,-15.20751351){\line(0,-1){5.77350269}}
\put(25,-15.20751351){\line(0,-1){5.77350269}}
\put(35,-15.20751351){\line(0,-1){5.77350269}}
\put(5,-15.20751351){\line(-1.732,1){5}}
\put(-5,-15.20751351){\line(-1.732,1){5}}
\put(-5,-15.20751351){\line(1.732,1){5}}
\put(5,-15.20751351){\line(1.732,1){5}}
\put(15,-15.20751351){\line(-1.732,1){5}}
\put(25,-15.20751351){\line(-1.732,1){5}}
\put(25,-15.20751351){\line(1.732,1){5}}
\put(15,-15.20751351){\line(1.732,1){5}}
\put(35,-15.20751351){\line(-1.732,1){5}}
\put(35,-15.20751351){\line(1.732,1){5}}

\put(0,-23.86785221){\line(0,-1){5.77350269}}
\put(10,-23.86785221){\line(0,-1){5.77350269}}
\put(20,-23.86785221){\line(0,-1){5.77350269}}
\put(30,-23.86785221){\line(0,-1){5.77350269}}
\put(10,-23.86785221){\line(-1.732,1){5}}
\put(0,-23.86785221){\line(-1.732,1){5}}
\put(0,-23.86785221){\line(1.732,1){5}}
\put(10,-23.86785221){\line(1.732,1){5}}
\put(30,-23.86785221){\line(-1.732,1){5}}
\put(20,-23.86785221){\line(-1.732,1){5}}
\put(20,-23.86785221){\line(1.732,1){5}}
\put(30,-23.86785221){\line(1.732,1){5}}

\put(15,-32.52819091){\line(-1.732,1){5}}
\put(5,-32.52819091){\line(-1.732,1){5}}
\put(5,-32.52819091){\line(1.732,1){5}}
\put(15,-32.52819091){\line(1.732,1){5}}
\put(25,-32.52819091){\line(-1.732,1){5}}
\put(25,-32.52819091){\line(1.732,1){5}}

\put(15,-9.43392605){\makebox(0,0){$0$}}
\put(5,-9.43392605){\makebox(0,0){$5$}}
\put(25,-9.43392605){\makebox(0,0){$2$}}
\put(10,-0.77358733){\makebox(0,0){$6$}}
\put(20,-0.77358733){\makebox(0,0){$1$}}
\put(10,-18.0942648){\makebox(0,0){$4$}}
\put(20,-18.0942648){\makebox(0,0){$3$}}

\thinlines

\put(60,-16.43392605){\makebox(0,0){$\overrightarrow{\quad\xi\quad}$}}

\put(80,-38.3014395){\line(1,0){34.6410162}}
\put(80,-26.7544341){\line(1,0){46.1880215}}
\put(80,-15.2074287){\line(1,0){57.7350269}}
\put(80.5470054,-3.66042332){\line(1,0){57.7350269}}
\put(91.5470054,7.88658206){\line(1,0){46.1880215}}
\put(103.094011, 19.4335874){\line(1,0){34.6410162}}

\put(80,-38.3014395){\line(0,1){34.6410162}}
\put(91.5470054,-38.3014395){\line(0,1){46.1880215}}
\put(103.094011,-38.3014395){\line(0,1){57.7350269}}
\put(114.641016,-38.3014395){\line(0,1){57.7350269}}
\put(126.188022,-26.7544341){\line(0,1){46.1880215}}
\put(137.735027,-15.2074287){\line(0,1){34.6410162}}

\put(108.867513,-9.43392605){\makebox(0,0){$0$}} \put(
97.3205081,-9.43392605){\makebox(0,0){$5$}}
\put(120.414519,-9.43392605){\makebox(0,0){$2$}}
\put(108.867513,2.11307933){\makebox(0,0){$6$}}
\put(97.3205081,-20.9809314){\makebox(0,0){$4$}}
\put(120.414519,2.11307933){\makebox(0,0){$1$}}
\put(108.867513,-20.9809314){\makebox(0,0){$3$}}
\end{picture}
\caption{The hexagonal model translation} \label{fig:hexmodel}
\end{figure}

\begin{lemma}
\label{lem:trans_lines} Two lines differ in length or slope in one
representation if and only if they differ in length or slope in
the other representation.
\end{lemma}
\begin{proof}
This claim can be verified easily by observing that two lines are
equal in length and slope in one representation if and only if
they are equal in length and slope in the other representation.
\end{proof}
\begin{cor}
A shape $\cS$ is a DDC in the hexagonal grid if and only if $\xi
(\cS) = \{ \xi ( p ) ~:~ p \in \cS \}$ is a DDC in $\Z^2$.
\end{cor}

Clearly, the representation of the hexagonal grid in terms of
$\Z^2$ implies that all the results on folding in the square grid
hold also in the hexagonal grid. We will consider now the most
important families of DDCs in the hexagonal grid, regular hexagons
and circles. A regular hexagon in the hexagonal grid is also
called an {\it hexagonal sphere} with radius $R$. It is a shape
with a center hexagon which includes all the points in the
hexagonal grid which are within Manhattan distance $R$ from the
center point. Applying the transformation $\xi$ on this sphere we
obtain a new shape in the square grid. This shape is a $(2R+1)
\times (2R+1)$ square from which isosceles right triangle with
sides of length $R$ are removed from the left upper corner and the
right lower corner. For the construction we use as our shape
$\cS$, in Theorem~\ref{thm:infinite}, a corner
CR($2R,w_1+w_2;R,w_2$), where $\frac{R}{w_2} \approx 1$, $|w_1 -
w_2 | \leq 3$ and $\text{g.c.d.}(w_1,w_2)=1$. In Appendix B a
construction for doubly periodic $\cS$-DDC, where $\cS$ is such
corner, is given where the number of dots in $\cS$ is
approximately $\sqrt{|\cS|} +o(\sqrt{|\cS|})$. By
Theorem~\ref{thm:trans_shape} the lattice tiling for $\cS$ is also
a lattice tiling for the shape $\cS'$ obtained from $\cS$ by
removing an isosceles right triangle with sides of length $R$ from
the lower left corner and adding it to the upper right corner of
the $\cS$ (see Figure~\ref{fig:corner_sphere}). The constructed
doubly periodic $\cS'$-DDC can be rotated by 90 degrees or flipped
either horizontally or vertically to obtain a doubly periodic
$\cQ$-DDC, where $\cQ$ is approximately an hexagonal sphere with
radius $R$. This yields a packing ratio approximately 1 between
the lower bound and the upper bound on the number of dots. Now, it
is easy to verify that the same construction, for a DDC with a
circle shape, given in Subsection~\ref{subsec:CircSqu} for the
square grid will work in the hexagonal grid. For this construction
we will use regular hexagon and a circle in the hexagonal grid to
obtain a packing ratio between the lower bound and the upper bound
on the number of dots in the circle which is the same as in the
square grid.

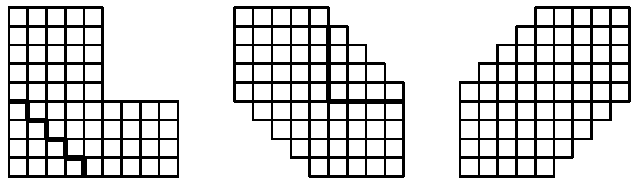
\begin{figure}[tb]

\setlength{\unitlength}{.5mm}
\begin{picture}(40,40)(40,0)
\linethickness{.5 pt}

\put(40,0){\line(1,0){45}} \put(40,5){\line(1,0){45}}
\put(40,10){\line(1,0){45}} \put(40,15){\line(1,0){45}}
\put(40,20){\line(1,0){45}} \put(40,25){\line(1,0){25}}
\put(40,30){\line(1,0){25}} \put(40,35){\line(1,0){25}}
\put(40,40){\line(1,0){25}} \put(40,45){\line(1,0){25}}

\put(40,0){\line(0,1){45}} \put(45,0){\line(0,1){45}}
\put(50,0){\line(0,1){45}} \put(55,0){\line(0,1){45}}
\put(60,0){\line(0,1){45}} \put(65,0){\line(0,1){45}}
\put(70,0){\line(0,1){20}} \put(75,0){\line(0,1){20}}
\put(80,0){\line(0,1){20}} \put(85,0){\line(0,1){20}}

\put(120,0){\line(1,0){25}} \put(115,5){\line(1,0){30}}
\put(110,10){\line(1,0){35}} \put(105,15){\line(1,0){40}}
\put(100,20){\line(1,0){45}} \put(100,25){\line(1,0){45}}
\put(100,30){\line(1,0){40}} \put(100,35){\line(1,0){35}}
\put(100,40){\line(1,0){30}} \put(100,45){\line(1,0){25}}

\put(100,20){\line(0,1){25}} \put(105,15){\line(0,1){30}}
\put(110,10){\line(0,1){35}} \put(115,5){\line(0,1){40}}
\put(120,0){\line(0,1){45}} \put(125,0){\line(0,1){45}}
\put(130,0){\line(0,1){40}} \put(135,0){\line(0,1){35}}
\put(140,0){\line(0,1){30}} \put(145,0){\line(0,1){25}}

\put(160,0){\line(1,0){25}} \put(160,5){\line(1,0){30}}
\put(160,10){\line(1,0){35}} \put(160,15){\line(1,0){40}}
\put(160,20){\line(1,0){45}} \put(160,25){\line(1,0){45}}
\put(165,30){\line(1,0){40}} \put(170,35){\line(1,0){35}}
\put(175,40){\line(1,0){30}} \put(180,45){\line(1,0){25}}

\put(160,0){\line(0,1){25}} \put(165,0){\line(0,1){30}}
\put(170,0){\line(0,1){35}} \put(175,0){\line(0,1){40}}
\put(180,0){\line(0,1){45}} \put(185,0){\line(0,1){45}}
\put(190,5){\line(0,1){40}} \put(195,10){\line(0,1){35}}
\put(200,15){\line(0,1){30}} \put(205,20){\line(0,1){25}}

\linethickness{1.5 pt} \put(125,20){\line(0,1){20}}
\put(125,20){\line(1,0){20}}

\put(40,20){\line(1,0){5}} \put(45,15){\line(1,0){5}}
\put(50,10){\line(1,0){5}} \put(55,5){\line(1,0){5}}
\put(45,20){\line(0,-1){5}} \put(50,15){\line(0,-1){5}}
\put(55,10){\line(0,-1){5}} \put(60,5){\line(0,-1){5}}

\end{picture}

\caption{From a corner CR$(9,9;5,4)$ to hexagonal sphere with
radius 4}\label{fig:corner_sphere}

\end{figure}

\section{Application for Error-Correction}
\label{sec:ECC}

In this section we will discuss the usage of folding to design
optimal (or "almost" optimal) codes which can correct adjacent
errors in a multidimensional array, i.e., a multidimensional
2-burst-correcting code. The construction is a generalization of
the construction of optimal one-dimensional 2-burst-correcting
codes given by Abramson~\cite{Abr59}. His construction was
generalized for larger bursts by~\cite{ElSh} and~\cite{AMOT} who
gave a comprehensive treatment for this topic. Multidimensional
generalization for the 2-burst-correcting codes were given
in~\cite{ScEt05,YaEt09}. We will give a multidimensional
generalization only for the 2-burst-correcting codes. The
parity-check matrix of a code of length $2^m-1$ and redundancy
$m+1$, consists of the $2^m-1$ consecutive nonzero elements
(powers of a primitive element $\alpha$) of GF($2^m$) followed by
a row of {\it ones}. The received word has one or two errors
depending if the last entry of its syndrome is {\it one} or {\it
zero}, respectively. The position of the error is determined by
the first $m$ entries of the syndrome.

The generalization of this idea is done by folding the nonzero
elements of GF($2^m$) into the parity-check matrix of a
multidimensional code row by row, dimension by dimension. Assume
that we have a $D$-dimensional array of size $n_1 \times n_2
\times \cdots \times n_D$ and we wish to correct any
$D$-dimensional burst of length 2 (at most two adjacent positions
are in error). The following construction given in~\cite{YaEt09}
is based on folding the nonzero elements of a Galois field with
characteristic 2 into a parity check matrix, where the order of
the elements of the field is determined by a primitive element of
the field.

\noindent {\bf Construction A:} Let $\alpha$ be a primitive
element in GF($2^m$) for the smallest integer $m$ such that $2^m-1
\geq \prod_{\ell=1}^D n_\ell$. Let $d=\lceil \log_2 D \rceil$ and
$\bi=(i_1,i_2,\ldots,i_D)$, where $0 \leq i_\ell \leq n_\ell-1$.
Let $A$ be a $d\times D$ matrix containing distinct binary
$d$-tuples as columns. We construct the following $n_1 \times n_2
\times \cdots \times n_D \times (m+d+1)$ parity check matrix $H$.

$$h_{\bi}=
\left[\begin{array}{c}
1 \\
A\bi^T \bmod 2 \\
\alpha^{\sum_{j=1}^D i_j (\prod_{\ell=j+1}^D n_\ell)}
\end{array}\right].
$$
for all $\bi=(i_1,i_2,\ldots,i_D)$, where $0 \leq i_\ell \leq
n_\ell-1$.

The following two theorems were given in~\cite{YaEt09}.
\begin{theorem}
The code constructed in Construction A can correct any 2-burst in
an $n_1 \times n_2 \times \cdots \times n_D$ array codeword.
\end{theorem}
\begin{theorem}
The code constructed by Construction A has redundancy which is
greater by at most one from the trivial lower bound on the
redundancy.
\end{theorem}

The same construction will work if instead of a $D$-dimensional
array our codewords will have have a shape $\cS$ of size $2^m-1$,
there is a lattice tiling $\Lambda$ for $\cS$, and there is a
direction vector $\delta$ such that $(\Lambda,\cS,\delta)$ defines
a folding. The nonzero elements of GF($2^m$) will be ordered along
the folded-row of $\cS$. Since usually the number of elements in
$\cS$ is not $2^m-1$ we should find a shape $\cS'$ which contains
$\cS$ and $| \cS' |=2^m-1$. We design a code with the shape of
$\cS'$ and since $\cS \subset \cS'$ the code will be able to
correct the same type of errors in $\cS$.

Finally, the construction can be generalized in a way that the
multidimensional code will be able to correct other types of two
errors in a multidimensional array~\cite{YaEt09}.

\section{Application for Pseudo-Random Arrays}
\label{sec:pseudo-random}

MacWilliams and Sloane~\cite{McSl76} gave the name {\it
pseudo-random sequence} to a maximal length sequence obtained from
a linear feedback shift register. These sequences called also PN
{Pseudo Noise} sequences or M-sequences have many desired
properties as described in~\cite{Golomb,McSl76}. The term
pseudo-random array was given by MacWilliams and
Sloane~\cite{McSl76} to a rectangular array obtained by folding a
pseudo-random sequence $S$ into its entries. The constructed
arrays can be obtained also as what is called maximum-area
matrices~\cite{NMIF}. In~\cite{McSl76} it was proved that if a
pseudo-random sequence of length $n=2^{k_1 k_2}-1$ is folded into
an $n_1 \times n_2$ array such that $n_1 = 2^{k_1}-1
>1$, $n_2 = \frac{n}{n_1} > 1$, and g.c.d$(n_1 , n_2) >1$ then the
constructed array has many desired properties and hence they
called this array $\cA$ a {\it pseudo-random array}. Some of the
properties they mentioned are as follows:
\begin{enumerate}
\item {\it Recurrences} - the entries satisfy a recurrence
relation along the folding.

\item {\it Balanced} - $2^{k_1 k_2 -1}$ entries in the array are
{\it ones} and $2^{k_1 k_2 -1}-1$ entries in the array are {\it
zeroes}.

\item {Shift-and-Add} - the sum of $\cA$ with any of its cyclic
shifts is another cyclic shift of $\cA$.

\item {\it Autocorrelation Function} - has two values: $n$
in-phase and -1 out-of-phase.

\item {\it Window property} - each of the $2^{k_1 k_2}-1$ nonzero
matrices of size $k_1 \times k_2$ is seen exactly once as a window
in the array.
\end{enumerate}

All these properties except for the window property are a
consequence of the fact that the elements in the folded-row are
consecutive elements of an M-sequence $S$. Before we examine
whether an array of any shape, obtained by folding $S$ into it,
has these properties we have to define what is a cyclic shift of
any given shape $\cS$ (even so we used the term without definition
before). Our definition will assume again that there exists a
lattice tiling $\Lambda$ for $\cS$ and a direction $\delta$ such
that $( \Lambda , \cS , \delta )$ defines a folding. A {\it cyclic
shift} of the shape $\cS$ (placed on the grid) is obtained by
taking the set of elements $\{ x + \delta ~:~ x \in \cS \}$.

\begin{lemma}
The shape of a cyclic shift of $\cS$ is $\cS$.
\end{lemma}
\begin{proof}
The cyclic shift is just a shift by $\delta$ of $\cS$ on the grid.
Therefore, the shape obtained is also $\cS$.
\end{proof}

\begin{theorem}
Let $\Lambda$ be a lattice tiling for a shape $\cS$ and let
$\delta$ be a direction such that $( \Lambda , \cS , \delta )$
defines a folding. If an M-sequence $S$ is folded into $\cS$ in
the direction $\delta$ then the Recurrences, Balanced,
Shift-and-Add, and the Autocorrelation Function properties hold
for the constructed array.
\end{theorem}
\begin{proof}
These properties follows immediately from the fact that the
entries of $\cS$ by the order of the folded-row are consecutive
elements of the M-sequence $S$. The two cyclic shifts of $\cS$
have the same folded-row up to a cyclic shift. Therefore, these
four properties are a direct consequence from the related
properties of the M-sequence.
\end{proof}

\begin{lemma}
\label{lem:related_window} Let $\Lambda$ be a lattice tiling for
the shape $\cS$ and $\delta$ be a direction for which the triple
$( \Lambda , \cS , \delta )$ defines a folding. Let $\cB$ be a
binary sequence of length $| \cS |$. Let $P_1$ and $P_2$ be two
points for which $P_1 - c(P_1) = P_2 - c(P_2)$. Then, for any two
positive integers $k_1$ and $k_2$ the two $k_1 \times k_2$ windows
of $( \Lambda , \cS , \delta, \cB )$ whose leftmost bottom points
are $P_1$ and $P_2$ are equal.
\end{lemma}
\begin{proof}
The lemma is an immediate consequence from the definition of the
lattice coloring induced by $( \Lambda , \cS , \delta )$ and the
definition of $( \Lambda , \cS , \delta, \cB )$.
\end{proof}

\begin{theorem}
\label{thm:window} Assume $\Lambda$ define a lattice tiling for an
$n_1 \times n_2$ array $\cA$, such that $n_1 n_2 = 2^{k_1k_2}-1$.
Assume further that $\Lambda$ defines a lattice tiling for the
shape $\cS$ and $(\Lambda,\cS,\delta)$ defines a folding for the
direction $\delta$. Then, if we fold an M-sequence $S$ into $\cS$
in the direction $\delta$, the resulting shape $\cS$ has the $k_1
\times k_2$ window property if and only if the $n_1 \times n_2$
array $\cA$ has the $k_1 \times k_2$ window property by folding
$S$ into $\cA$ in the direction $\delta$.
\end{theorem}
\begin{proof}
Since $\Lambda$ is a lattice tiling for both $\cA$ and $\cS$ there
is a sequence of arrays $\cA_0 = \cA$, $\cA_1$,...,$\cA_r = \cS$,
such that $| \cA_{i+1} \setminus \cA_i | = | \cA_i \setminus
\cA_{i+1} | =1$, $0 \leq i \leq r-1$, $\Lambda$ is a lattice
tiling for $\cA_i$, $0 \leq i \leq r$, and the origin is contained
in $\cA_i$, $0 \leq i \leq r$. Moreover, it is easy to verify that
given the shape $\cA_i$, $P_1 = \cA_{i+1} \setminus \cA_i$, $P_2 =
\cA_i \setminus \cA_{i+1}$, we have that $P_2 = P_1 -c(P_1)$ with
respect to $\cA_i$. The theorem follows now by induction and using
Lemma~\ref{lem:related_window}.
\end{proof}

Theorem~\ref{thm:window} does not give any new information about
window sizes which are not covered in~\cite{McSl76,NMIF}. The
following lemma provides such information. We say that a shape
$\cS$ of size $2^n-1$ has the {$\cQ$ window property} if $|\cQ|=n$
and each nonzero value for $\cQ$ appears exactly once in a copy of
$\cS$, where $\cS$ is considered to be a cyclic shape.

\begin{lemma}
Let $\Lambda$ be a lattice tiling for a shape $\cS$,
$|\cS|=2^n-1$, $\delta$ be a direction vector, and $S$ be an
M-sequence of length $2^n-1$. Let $\cQ$ be a shape with volume
$n$. If in the array $\cS'$ defined by $(\Lambda , \cS , \delta ,
S)$ there is no copy of $\cQ$ which contains only {\it zeroes}
then $\cS$ has the $\cQ$ window property.
\end{lemma}
\begin{proof}
By the Shift-and-Add property, $\cS'$ has two identical copies of
$\cQ$ if and only if $\cS'$ has a copy of $\cQ$ which contains
only {\it zeroes}. Thus, $\cS'$ has the $\cQ$ window property if
and only if there is no copy of $\cQ$ in $\cS'$ which contains
only {\it zeroes}.
\end{proof}
We can use now the properties we have found for the generalized
folding to obtain various results. An example is given in the
following corollary.
\begin{cor}
Let $\Lambda$ be a lattice tiling for a shape $\cS$,
$|\cS|=2^n-1$, and $S$ be an M-sequence of length $2^n-1$. If
$2^n-1$ is a Mersenne prime then $(\Lambda,\cS,\delta,S)$ has the
$1 \times n$ and the $n \times 1$ window property for any given
direction vector $\delta$.
\end{cor}

\begin{example}
Consider the following M-sequence
$S=0000100101100111110001101110101$ of length 31. Let $\Lambda$ be
a lattice tiling for a corner CR(5,7;1,4) with the generator
matrix
$$
G_2=\left[\begin{array}{cc}
3 & 4 \\
10 & 3
\end{array}\right]~.
$$
By folding of $S$ in the direction $(+1,0)$ we obtain the
following pseudo-random array

\vspace{0.8cm}

\setlength{\unitlength}{.75mm}
\begin{picture}(40,25)(10,-5)
\linethickness{.5 pt}

\put(30,0){\framebox(35,15){}} \put(30,5){\framebox(35,5){}}
\put(30,10){\framebox(35,10){}}

\put(50,0){\framebox(15,20){}} \put(55,0){\framebox(5,20){}}
\put(30,0){\framebox(15,25){}} \put(35,0){\framebox(5,25){}}

\put(30,0){\makebox(5,5){0}} \put(35,0){\makebox(5,5){0}}
\put(40,0){\makebox(5,5){0}} \put(45,0){\makebox(5,5){0}}
\put(50,0){\makebox(5,5){1}} \put(55,0){\makebox(5,5){0}}
\put(60,0){\makebox(5,5){0}}

\put(30,5){\makebox(5,5){1}} \put(35,5){\makebox(5,5){0}}
\put(40,5){\makebox(5,5){1}} \put(45,5){\makebox(5,5){1}}
\put(50,5){\makebox(5,5){0}} \put(55,5){\makebox(5,5){0}}
\put(60,5){\makebox(5,5){1}}

\put(30,10){\makebox(5,5){1}} \put(35,10){\makebox(5,5){1}}
\put(40,10){\makebox(5,5){1}} \put(45,10){\makebox(5,5){1}}
\put(50,10){\makebox(5,5){0}} \put(55,10){\makebox(5,5){0}}
\put(60,10){\makebox(5,5){0}}

\put(30,15){\makebox(5,5){1}} \put(35,15){\makebox(5,5){1}}
\put(40,15){\makebox(5,5){0}} \put(45,15){\makebox(5,5){1}}
\put(50,15){\makebox(5,5){1}} \put(55,15){\makebox(5,5){1}}
\put(60,15){\makebox(5,5){0}}

\put(30,20){\makebox(5,5){1}} \put(35,20){\makebox(5,5){0}}
\put(40,20){\makebox(5,5){1}}

\end{picture}

This array has the $5 \times 1$ and $1 \times 5$ window
properties. Out of the 19 shapes of size 5 with exactly two rows
it does not have the window property only for the following three
shapes:

\vspace{-0.2cm}

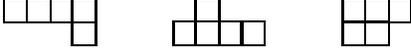
\setlength{\unitlength}{.75mm}
\begin{picture}(40,25)(10,-5)
\linethickness{.5 pt}

\put(20,4){\framebox(16,4){}} \put(24,4){\framebox(4,4){}}
\put(32,0){\framebox(4,8){}}

\put(50,0){\framebox(12,4){}} \put(54,0){\framebox(4,8){}}
\put(62,0){\framebox(4,4){}}

\put(80,0){\framebox(8,4){}} \put(80,4){\framebox(12,4){}}
\put(84,0){\framebox(4,8){}}

\end{picture}

The pseudo-random array obtained by folding $S$ by the direction
$(0,+1)$ is

\vspace{0.8cm}

\setlength{\unitlength}{.75mm}
\begin{picture}(40,25)(10,-5)
\linethickness{.5 pt}

\put(30,0){\framebox(35,15){}} \put(30,5){\framebox(35,5){}}
\put(30,10){\framebox(35,10){}}

\put(50,0){\framebox(15,20){}} \put(55,0){\framebox(5,20){}}
\put(30,0){\framebox(15,25){}} \put(35,0){\framebox(5,25){}}

\put(30,0){\makebox(5,5){0}} \put(35,0){\makebox(5,5){1}}
\put(40,0){\makebox(5,5){0}} \put(45,0){\makebox(5,5){0}}
\put(50,0){\makebox(5,5){0}} \put(55,0){\makebox(5,5){1}}
\put(60,0){\makebox(5,5){0}}

\put(30,5){\makebox(5,5){0}} \put(35,5){\makebox(5,5){1}}
\put(40,5){\makebox(5,5){0}} \put(45,5){\makebox(5,5){1}}
\put(50,5){\makebox(5,5){0}} \put(55,5){\makebox(5,5){1}}
\put(60,5){\makebox(5,5){1}}

\put(30,10){\makebox(5,5){0}} \put(35,10){\makebox(5,5){0}}
\put(40,10){\makebox(5,5){0}} \put(45,10){\makebox(5,5){0}}
\put(50,10){\makebox(5,5){1}} \put(55,10){\makebox(5,5){1}}
\put(60,10){\makebox(5,5){1}}

\put(30,15){\makebox(5,5){0}} \put(35,15){\makebox(5,5){0}}
\put(40,15){\makebox(5,5){1}} \put(45,15){\makebox(5,5){1}}
\put(50,15){\makebox(5,5){0}} \put(55,15){\makebox(5,5){1}}
\put(60,15){\makebox(5,5){1}}

\put(30,20){\makebox(5,5){1}} \put(35,20){\makebox(5,5){1}}
\put(40,20){\makebox(5,5){1}}

\end{picture}

It has the $5 \times 1$ and $1 \times 5$ window properties. But,
out of the 19 shapes of size 5 with exactly two rows it does not
have the window property for eight shapes.

Both pseudo-random arrays have a window property for the star
shape given by

\vspace{-0.4cm}

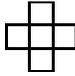
\setlength{\unitlength}{.75mm}
\begin{picture}(40,25)(10,-5)
\linethickness{.5 pt}

\put(55,0){\framebox(12,4){}} \put(59,-4){\framebox(4,12){}} ~.

\end{picture}

\end{example}

\section{Conclusion and Open Problems}
\label{sec:conclude}

The well-known definition of folding was generalized. The
generalization and its applications led to several new results
summarized as follows:

\begin{enumerate}
\item The generalization is based on a lattice tiling for a shape
$\cS$ and a direction $\delta$. The number of possible
nonequivalent directions is $\frac{\mu ( | \cS |)}{2}$. Necessary
and sufficient conditions that a direction defines a folding are
derived.

\item Folding a $B_2$-sequence into a shape $\cS$ result in a
distinct difference configuration with the shape $\cS$.

\item Lower bounds on the number of dots in a distinct difference
configuration with shape of regular polygon, circle, and other
interesting geometrical shapes are derived.

\item Low redundancy multidimensional codes for correcting a burst
of length two are obtained.

\item New pseudo-random arrays with window and correlation
properties are derived. These arrays differ from known arrays
either in their shape or the shape of their window property.
\end{enumerate}

The discussion on these results leads to many new interesting open
problems. We conclude with a list of six open problems related to
our discussion.

\begin{enumerate}
\item We have discussed several applications for the folding
operation in general and for the new generalization of folding in
particular. We believe that there are more interesting
applications for this operation and we would like to see them
explored.

\item The construction for DDCs whose shape is a quasi-perfect
hexagon works for infinite number of parameters. But, the set of
parameters is very sparse. Its density depends on the number of
primes obtained by Dirichlet's Theorem. This immediately implies
the same for the parameters of DDCs whose shape is a regular
polygon. We would like to see a construction of such DDCs with a
dense set of parameters.

\item What is the lower bound on the number of dots in a DDC whose
shape is a circle with radius $R$? We conjecture that the lower
bound is $\sqrt{\pi}R +o(R)$.

\item We would like to see an asymptotic improvement on the lower
bounds on the number of dots in a DDC whose shape is a regular
$n$-gon with radius $R$.

\item Are there cases where we can improve the upper bound on the
number of dots in these DDCs asymptotically?

\item We would like to see a more general theorem which connects
folding of M-sequences and general window property.
\end{enumerate}

\section*{Appendix A}

In this Appendix we prove the necessary and sufficient condition
for a triple $(\Lambda , \cS , \delta )$ to define a folding. For
the proof of the theorem we use the well known Cramer's
rule~\cite{MacBi88} which is given first.

\begin{theorem}
\label{thm:cramer} Given the  following system with the $n$ linear
equations and the variables $x_1 , x_2 , \ldots , x_n$
$$
\left[\begin{array}{cccc}
a_{11} & a_{12} & \ldots & a_{1n} \\
a_{21} & a_{22} & \ldots & a_{2n} \\
\vdots & \vdots & \ddots & \vdots\\
a_{n1} & a_{n2} & \ldots & a_{nn} \end{array}\right]
\left[\begin{array}{c} x_1 \\ x_2 \\ \vdots \\ x_n \end{array}
\right] = \left[\begin{array}{c} b_1 \\ b_2 \\ \vdots \\ b_n
\end{array} \right] ~.
$$
If
$$
A= \det \left|\begin{array}{cccc}
a_{11} & a_{12} & \ldots & a_{1n} \\
a_{21} & a_{22} & \ldots & a_{2n} \\
\vdots & \vdots & \ddots & \vdots\\
a_{n1} & a_{n2} & \ldots & a_{nn} \end{array}\right| ~,
$$
then $x_k = \frac{A_k}{A}$ for $1 \leq k \leq n$, where
$$
A_k = \det \left|\begin{array}{ccccccc}
a_{11} & \ldots & a_{1(k-1)} & b_1 & a_{1(k+1)} & \ldots & a_{1n} \\
a_{21} & \ldots & a_{2(k-1)} & b_2 & a_{2(k+1)} & \ldots & a_{2n} \\
\vdots & \ddots & \vdots & \dots & \vdots & \ddots & \vdots\\
a_{n1} & \ldots & a_{n(k-1)} & b_n & a_{n(k+1)} & \ldots & a_{nn}
\end{array}\right| ~.
$$
\end{theorem}

\vspace{1.3cm}

Let $\Lambda$ be a $D$-dimensional lattice tiling for the shape
$\cS$. Let $G$ be the following generator matrix of $\Lambda$:

$$
G=\left[\begin{array}{cccc}
v_{11} & v_{12} & \ldots & v_{1D} \\
v_{21} & v_{22} & \ldots & v_{2D} \\
\vdots & \vdots & \ddots & \vdots\\
v_{D1} & v_{D2} & \ldots & v_{DD} \end{array}\right]~.
$$

Given the direction vector $\delta = (d_1 ,d_2, \ldots , d_D)$,
w.l.o.g. we assume that the first $\ell_1 \geq 1$ values of
$\delta$ are positives, the next $\ell_2$ values are negatives,
and the last $D-\ell_1-\ell_2$ values are 0's. By
Lemma~\ref{lem:tile_fold2} and Corollary~\ref{cor:lattice_points},
if $(\Lambda,\cS,\delta)$ defines a folding then there exist $D$
integer coefficients $\alpha_1,\alpha_2,\ldots,\alpha_D$ such that
\begin{eqnarray*}
\sum_{j=1}^D \alpha_j
(v_{j1},v_{j2},\ldots,v_{jD})=~~~~~~~~~~~~~~~~~~~~~~~~~~\\
(|\cS| d_1 ,\ldots,|\cS| d_{\ell_1},-|\cS| d_{\ell_1
+1},\ldots,-|\cS| d_{\ell_1 + \ell_2},0,\ldots,0),
\end{eqnarray*}
and there is no integer $i$, $0 < i < |\cS|$, and $D$ integer
coefficients $\beta_1,\beta_2,\ldots,\beta_D$ such that
\begin{eqnarray*}
\sum_{j=1}^D \beta_j
(v_{j1},v_{j2},\ldots,v_{jD})~~~~~~~~~~~~~~~~~~\\
=(i \cdot d_1,\ldots,i \cdot d_{\ell_1} ,-i \cdot d_{\ell_1
+1},\ldots,-i \cdot d_{\ell_1 + \ell_2},0,\ldots,0)~.
\end{eqnarray*}
Hence we have the following $D$ equations:
\begin{equation}
\label{eq:Gset+1} \sum_{j=1}^D \alpha_j v_{jr} = |\cS| \cdot d_r
,~~~ 1 \leq r \leq \ell_1 ,
\end{equation}
\begin{equation}
\label{eq:Gset-1} \sum_{j=1}^D \alpha_j v_{jr} = -|\cS| \cdot
d_r,~~~ \ell_1 +1 \leq r \leq \ell_1 + \ell_2 ,
\end{equation}
\begin{equation}
\label{eq:Gset0} \sum_{j=1}^D \alpha_j v_{jr} = 0,~~~
\ell_1+\ell_2+1 \leq r \leq D ~.
\end{equation}
Let $\tau = d_1$ if $\ell_1 + \ell_2 =1$ and $\tau =
\text{g.c.d.}(d_1 , d_2, \ldots , d_{\ell_1 + \ell_2} )$ if
$\ell_1 + \ell_2 >1$. The $D$ equations in (\ref{eq:Gset+1}),
(\ref{eq:Gset-1}), (\ref{eq:Gset0}) are equivalent to the
following $D$ equations:
$$
\sum_{j=1}^D \alpha_j v_{j1} = |\cS| \cdot d_1,
$$
\begin{equation*}
\label{eq:Gset+1a} \sum_{j=1}^D \alpha_j \frac{d_1 v_{jr}-d_r
v_{j1}}{\tau} = 0,~~~ 2 \leq r \leq \ell_1 ,
\end{equation*}
\begin{equation*}
\label{eq:Gset-1a} \sum_{j=1}^D \alpha_j \frac{d_1 v_{jr}+d_r
v_{j1}}{\tau} = 0,~~~ \ell_1 +1 \leq r \leq \ell_1 + \ell_2,
\end{equation*}
\begin{equation*}
\label{eq:Gset0a} \sum_{j=1}^D \alpha_j v_{jr} = 0,~~~
\ell_1+\ell_2+1 \leq r \leq D .
\end{equation*}
We define now a set of $D(D-1)$ new coefficients $u_{rj}$, $2 \leq
r \leq D$, $1 \leq j \leq D$, as follows:

$$
u_{rj}=\frac{d_1 v_{jr}- d_r v_{j1}}{\tau}~~ \text{for}~ 2 \leq r
\leq \ell_1 ,
$$
$$
u_{rj}=\frac{d_1 v_{jr}+ d_r v_{j1}}{\tau}~~ \text{for}~ \ell_1+1
\leq r \leq \ell_1+\ell_2 ,
$$
$$
u_{rj}=v_{jr}~~ \text{for}~ \ell_1+\ell_2+1 \leq r \leq D .
$$
Consider the $(D-1) \times D$ matrix
$$
H=\left[\begin{array}{cccc}
u_{21} & u_{22} & \ldots & u_{2D} \\
u_{31} & u_{32} & \ldots & u_{3D} \\
\vdots & \vdots & \ddots & \vdots\\
u_{D1} & u_{D2} & \ldots & u_{DD} \end{array}\right] ~.
$$
Using Theorem~\ref{thm:cramer} it is easy to verify that the
unique solution for the $\alpha_k$'s is
\begin{equation}
\label{eq:alpha} \alpha_k = (-1)^{k-1} \frac{d_1 \tau^{\ell_1 +
\ell_2 -1} \det H_k }{d_1^{\ell_1 + \ell_2 -1}}
\end{equation}
%
where $H_k$ is the $(D-1) \times (D-1)$ matrix obtained from $H$
by deleting column $k$ of $H$.

\begin{lemma}
For each $k$, $1 \leq k \leq D$, $\tau$ divides $\alpha_k$ defined
in (\ref{eq:alpha}).
\end{lemma}
\begin{proof}
Consider the  following $D \times D$ matrix
$$
\tilde{G}=\left[\begin{array}{cccc}
v_{11} & v_{21} & \ldots & v_{D1} \\
u_{21} & u_{22} & \ldots & u_{2D} \\
u_{31} & u_{32} & \ldots & u_{3D} \\
\vdots & \vdots & \ddots & \vdots\\
u_{D1} & u_{D2} & \ldots & u_{DD} \end{array}\right] ~.
$$
By the definition of the entries in the matrix $H$ and since $\det
G = |\cS|$ it follows that that $\det \tilde{G} = |\cS| \left(
\frac{d_1}{\tau} \right)^{\ell_1+\ell_2-1}$. $\det \tilde{G}$ in
Theorem~\ref{thm:cramer} is equal $A$, while $A_k$ is equal $|\cS|
\cdot d_1 \left( \frac{d_1}{\tau} \right)^{\ell_1+\ell_2-2} Y$,
for some integer $Y$. Therefore, $\alpha_k = \tau Y$ and the lemma
follows.
\end{proof}

This analysis leads to the following theorem.
\begin{theorem}
\label{thm:new_general_condition} If $\Lambda$ is a lattice tiling
for the shape $\cS$ then the triple $(\Lambda,\cS,\delta)$ defines
a folding if and only if $\text{g.c.d.}(\frac{\alpha_1}{\tau} ,
\frac{\alpha_2}{\tau} , \ldots , \frac{\alpha_D}{\tau})=1$ and
$\text{g.c.d.}( \tau , | \cS | )=1$.
\end{theorem}
\begin{proof}
Assume first that $(\Lambda,\cS,\delta)$ defines a folding.

Now, assume for the contrary that
$\text{g.c.d.}(\frac{\alpha_1}{\tau} , \frac{\alpha_2}{\tau} ,
\ldots , \frac{\alpha_D}{\tau})= \nu_1 > 1$ or $\text{g.c.d.}(
\tau , | \cS | )= \nu_2 > 1$. We distinguish between two cases.

\noindent {\bf Case 1:}

Assume that $\text{g.c.d.}(\frac{\alpha_1}{\tau} ,
\frac{\alpha_2}{\tau} , \ldots , \frac{\alpha_D}{\tau})= \nu_1
>1$.

Equations (\ref{eq:Gset+1}), (\ref{eq:Gset-1}), and
(\ref{eq:Gset0}) have exactly one solution for the $\alpha_i$'s
given in (\ref{eq:alpha}). Since
$\text{g.c.d.}(\frac{\alpha_1}{\tau} , \frac{\alpha_2}{\tau} ,
\ldots , \frac{\alpha_D}{\tau})= \nu_1$, it follows that $\beta_i
= \frac{\alpha_i}{\tau \nu_1}$, $1 \leq i \leq D$, are integers.
Therefore, we have

\begin{equation*}
\label{eq:Gbset+1} \sum_{j=1}^D \beta_j v_{jr} = \frac{|\cS|}{\tau
\nu_1} d_r,~~~ 1 \leq r \leq \ell_1 ,
\end{equation*}
\begin{equation*}
\label{eq:Gbset-1} \sum_{j=1}^D \beta_j v_{jr} =
\frac{-|\cS|}{\tau \nu_1} d_r,~~~ \ell_1 +1 \leq r \leq \ell_1 +
\ell_2 ,
\end{equation*}
\begin{equation*}
\label{eq:Gbset0} \sum_{j=1}^D \beta_j v_{jr} = 0,~~~
\ell_1+\ell_2+1 \leq r \leq D ,
\end{equation*}
i.e.,
\begin{eqnarray*}
\sum_{j=1}^D \beta_j
(v_{j1},v_{j2},\ldots,v_{jD})=~~~~~~~~~~~~~~~~~~~~~~~~~~\\
(\frac{|\cS|}{\tau \nu_1} d_1,\ldots,\frac{|\cS|}{\tau \nu_1}
d_{\ell_1},-\frac{|\cS|}{\tau \nu_1}
d_{\ell_1+1},\ldots,-\frac{|\cS|}{\tau \nu_1} d_{\ell_1
+\ell_2},0,\ldots,0),
\end{eqnarray*}
and as a consequence by Lemma~\ref{lem:tile_fold2} we have that
$(\Lambda,\cS,\delta)$ does not define a folding, a contradiction.

\noindent {\bf Case 2:}

Assume that $\text{g.c.d.}( \tau , | \cS | )= \nu_2 >1$.

Let $\beta_i = \frac{\alpha_i}{\nu_2}$, $1 \leq i \leq \ell_1 +
\ell_2$. Therefore,
\begin{eqnarray*}
\sum_{j=1}^D \beta_j
(v_{j1},v_{j2},\ldots,v_{jD})=~~~~~~~~~~~~~~~~~~~~~~~~~~\\
(\frac{|\cS|}{\nu_2} d_1,\ldots,\frac{|\cS|}{\nu_2}
d_{\ell_1},-\frac{|\cS|}{\nu_2}
d_{\ell_1+1},\ldots,-\frac{|\cS|}{\nu_2} d_{\ell_1
+\ell_2},0,\ldots,0),
\end{eqnarray*}
and as a consequence by Lemma~\ref{lem:tile_fold2} we have that
$(\Lambda,\cS,\delta)$ does not define a folding, a contradiction.

\vspace{0.15cm}

As a consequence of Case 1 and Case 2 we have that if
$(\Lambda,\cS,\delta)$ defines a folding with the ternary vector
$\delta$ then $\text{g.c.d.}(\frac{\alpha_1}{\tau} ,
\frac{\alpha_2}{\tau} , \ldots , \frac{\alpha_D}{\tau})=1$ and
$\text{g.c.d.}( \tau , | \cS | )=1$.

\vspace{0.3cm}

Now assume that $\text{g.c.d.}(\frac{\alpha_1}{\tau} ,
\frac{\alpha_2}{\tau} , \ldots , \frac{\alpha_D}{\tau})=1$ and
$\text{g.c.d.}( \tau , | \cS | )=1$. Consider the set of $D$
equations defined by
\begin{equation}
\label{eq:Gset_alpha} \sum_{j=1}^D \alpha_j
(v_{j1},v_{j2},\ldots,v_{jD})=~~~~~~~~~~~~~~~~~~~~~~~~~~
\end{equation}
\begin{equation*}
(|\cS| d_1,\ldots,|\cS| d_{\ell_1},-|\cS|
d_{\ell_1+1},\ldots,-|\cS| d_{\ell_1+\ell_2},0,\ldots,0),
\end{equation*}
Since the rows of $G$ are linearly independent, it follows that
this set of equations has a unique solution for the $\alpha_i$'s
(but, these coefficients are not necessary integers). Using the
same analysis proceeding the theorem, we have by the Cramer's rule
that this solution is given by (\ref{eq:alpha}) and hence the
$\alpha_i$'s are integers. Assume for the contrary that
$(\Lambda,\cS,\delta)$ does not define a folding. Then, by
Lemma~\ref{lem:tile_fold2} we have that there exist $D$ integers
$\beta_i$, $1 \leq i \leq D$, such that
\begin{equation}
\label{eq:Gset_beta} \sum_{j=1}^D \beta_j
(v_{j1},v_{j2},\ldots,v_{jD})=~~~~~~~~~~~~~~~~~~~~~~~~~~
\end{equation}
\begin{equation*}
(\ell \cdot d_1 ,\ldots,\ell \cdot d_{\ell_1},-\ell \cdot
d_{\ell_1 +1} ,\ldots,-\ell \cdot d_{\ell_1 + \ell_2}
,0,\ldots,0), \label{eq:set_beta}
\end{equation*}
for some integer $0 < \ell < |\cS|$.

Since the rows of $G$ are linearly independent then there exists
exactly one set of $\beta_i$'s (integers or non-integers) which
satisfies (\ref{eq:Gset_beta}). Let $\nu
=\text{g.c.d.}(\ell,|\cS|)$, where clearly $1 \leq \nu \leq \ell <
|\cS|$. From equations (\ref{eq:Gset_alpha}) and
(\ref{eq:Gset_beta}) we obtain

\begin{eqnarray*}
\sum_{j=1}^D (\ell \alpha_j )
(v_{j1},v_{j2},\ldots,v_{jD})=~~~~~~~~~~~~~~~~~~~~~~~~~~\\
(\ell |\cS| d_1,\ldots,\ell |\cS| d_{\ell_1},-\ell |\cS|
d_{\ell_1+1},\ldots,-\ell |\cS| d_{\ell_1+\ell_2},0,\ldots,0) \\
=\sum_{j=1}^D (|\cS| \beta_j )
(v_{j1},v_{j2},\ldots,v_{jD})~,~~~~~~~~~~~~~~
\end{eqnarray*}
Since the rows of $G$ are linearly independent it implies that
$\ell \alpha_i = |\cS| \beta_i$ for each $1 \leq i \leq D$, i.e.,
$\beta_i = \frac{\ell \alpha_i}{|\cS|}$. $\beta_i = \frac{\ell
\alpha_i}{|\cS|}$ is an integer and $\nu
=\text{g.c.d.}(\ell,|\cS|)$ implies that $\beta_i = \frac{\ell /
\nu}{|\cS| / \nu} \alpha_i$, $1 \leq i \leq D$.
$\text{g.c.d.}(\ell / \nu,|\cS| / \nu)=1$ and hence
$\frac{|\cS|}{\nu}$ divides $\alpha_i$ for each $i$, $1 \leq i
\leq D$. $\text{g.c.d.}( \tau , | \cS | )=1$, $\tau$ divides
$\alpha_i$, and hence $\frac{|\cS|}{\nu}$ divides
$\frac{\alpha_i}{\tau}$ for each $i$, $1 \leq i \leq D$. Hence,
$\text{g.c.d.}(\frac{\alpha_1}{\tau} , \frac{\alpha_2}{\tau} ,
\ldots , \frac{\alpha_D}{\tau}) \geq \frac{|\cS|}{\nu}$. But,
$\text{g.c.d.}(\frac{\alpha_1}{\tau} , \frac{\alpha_2}{\tau} ,
\ldots , \frac{\alpha_D}{\tau})=1$ and hence $\nu =|\cS|$, i.e.,
$\ell \geq | \cS |$, a contradiction. Thus, $(\Lambda,\cS,\delta)$
defines a folding.
\end{proof}

\section*{Appendix B}

In this appendix we consider DDCs with two special shapes, called
corner and flipped T. The DDCs with these shapes and special
parameters are important in applying Theorem~\ref{thm:infinite} to
obtain other DDCs such as triangles in the square grid and
hexagonal spheres in the hexagonal grid.

\subsection{Corner}

A {\it corner}, CR$(h_1+h_2,w_1+w_2;h_2,w_2)$, is an $(h_1 + h_2)
\times (w_1 + w_2)$ rectangle from which an $h_2 \times w_2$
rectangle was removed from its right upper corner. An example is
given in Figure~\ref{fig:corner}. Let $\cS$ be a
CR$(h_1+h_2,w_1+w_2;h_2,w_2)$ and let $\Lambda$ the lattice with
the following generator matrix
$$
G=\left[\begin{array}{cc}
w_1 & h_1 \\
-w_2 & h_1 +h_2
\end{array}\right] ~.
$$

\begin{figure}[tb]
\centering

\setlength{\unitlength}{.5mm}
\begin{picture}(30,40)(40,0)
\linethickness{.5 pt}

\put(30,0){\framebox(55,5){}} \put(30,10){\framebox(55,5){}}
\put(30,20){\framebox(55,5){}} \put(30,30){\framebox(35,5){}}
\put(30,0){\framebox(35,35){}} \put(35,0){\framebox(25,35){}}
\put(40,0){\framebox(15,35){}} \put(45,0){\framebox(5,35){}}
\put(70,0){\framebox(15,25){}} \put(75,0){\framebox(5,25){}}

\end{picture}

\caption{A corner CR$(7,11;2,4)$}\label{fig:corner}

\end{figure}
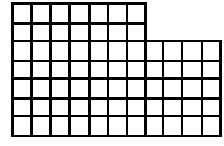

Clearly, $\Lambda$ is a lattice tiling for $\cS$. A general result
concerning DDCs whose shape is a corner seems to be quite
difficult. We will consider the case which seems to be the most
useful for our purpose. First note, that by
Theorem~\ref{thm:cond_fold2D}, $\delta = (0,+1)$ defines a folding
for $\Lambda$ if and only if $\text{g.c.d.}(w_1,w_2)=1$. Assume
first that $h_1=h_2$ and $|w_1-w_2| \leq 3$. By
Theorem~\ref{thm:ratio_rec}, we have an $n_1 \times n_2$ rectangle
$\cQ$ such that $n_1 n_2 = p^2 -1$ for some prime $p$,
$\frac{2n_1}{n_2} \approx \frac{h_1+h_2}{2w_1+w_2}$, and $n_1$ is
even. Now, we will make new choices for $h_1$, $h_2$, $w_1$, and
$w_2$, which are close to the old ones. Let $h_1=h_2= n_1$; we
distinguish between three cases of $n_2$:
\begin{enumerate}
\item [(W.1)] If $n_2 = 3 \omega+1$ then $w_1 = \omega$ and
$w_2=\omega+1$.

\item [(W.2)] If $n_2 = 3 \omega+2$ then $w_1 = \omega +1$ and
$w_2=\omega$.

\item [(W.3)] If $n_2 = 3 \omega$ then we distinguish between two
cases:
\begin{itemize}
\item if $\omega -1 \equiv 0~(\bmod 3)$ then $w_1 = \omega+1$ and
$w_2=\omega -2$.

\item  if $\omega -1 \not\equiv 0~(\bmod 3)$ then $w_1 = \omega-1$
and $w_2=\omega +2$.
\end{itemize}
\end{enumerate}

It is easy to verify that the size of the new corner
CR$(h_1+h_2,w_1+w_2;h_2,w_2)$, $\cS'$, is $n_1 n_2 = p^2-1$,
$\Lambda$ is a lattice tiling for $\cS'$, $(\Lambda,\cS',\delta)$,
$\delta = (0,+1)$, defines a folding, and we can form a doubly
periodic $\cS'$-DDC with it. Hence, we have the following theorem.

\begin{theorem}
Let $n_1$ and $n_2$ be two integers such that $n_1 n_2 = p^2 -1$
for some prime number $p$, $n_2 = 2w_1+w_2$, where $n_1$ is an
even integer, $w_1$, $w_2$, are defined by (W.1), (W.2), (W.3).
Then there exists a doubly periodic $\cS$-DDC, whose shape is a
corner, CR$(2 n_1,w_1+w_2;n_1,w_2)$, with $p$ dots.
\end{theorem}

\subsection{Flipped T}

\begin{figure}[tb]
\centering

\setlength{\unitlength}{.5mm}
\begin{picture}(30,50)(40,0)
\linethickness{.5 pt}

\put(15,0){\framebox(85,5){}} \put(15,10){\framebox(85,5){}}
\put(15,20){\framebox(85,5){}} \put(15,0){\framebox(15,25){}}
\put(20,0){\framebox(5,25){}} \put(75,0){\framebox(25,25){}}
\put(80,0){\framebox(15,25){}} \put(85,0){\framebox(5,25){}}
\put(35,30){\framebox(35,5){}} \put(35,40){\framebox(35,5){}}
\put(35,0){\framebox(35,50){}} \put(40,0){\framebox(25,50){}}
\put(45,0){\framebox(15,50){}} \put(50,0){\framebox(5,50){}}

\end{picture}

\caption{A flipped T FT$(5,17;4,6)$}\label{fig:flippedT}

\end{figure}
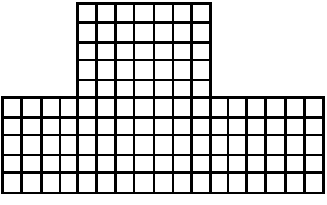

A {\it flipped T}, FT$(h,w_1+w_2+w_3;w_1,w_3)$, is an $(2h) \times
(w_1 + w_2+w_3)$ rectangle from which an $h \times w_1$ rectangle
was removed from its left upper corner and an $h \times w_3$
rectangle was removed from its right upper corner. An example is
given in Figure~\ref{fig:flippedT}. Let $\cS$ be a
FT$(h,w_1+w_2+w_3;w_1,w_3)$ and let $\Lambda$ the lattice with the
following generator matrix
$$
G=\left[\begin{array}{cc}
w_1+w_2 & h \\
w_1+2w_2+w_3 & 0
\end{array}\right]~.
$$

Clearly, $\Lambda$ is a lattice tiling for $\cS$. A general result
concerning DDCs whose shape is a flipped T seems to be quite
difficult. We will consider the case which seems to be the most
useful for our purpose. First note, that by
Theorem~\ref{thm:cond_fold2D}, $\delta = (0,+1)$ defines a folding
for $\Lambda$ if and only if
$\text{g.c.d.}(w_1+w_2,w_1+2w_2+w_3)=1$ which is equivalent to
$\text{g.c.d.}(w_1+w_2,w_2+w_3)=1$. Assume that $|w_1-w_3| \leq
4$. By Theorem~\ref{thm:ratio_rec}, we have an $n_1 \times n_2$
rectangle $\cQ$ such that $n_1 n_2 = p^2 -1$ for some prime $p$,
$\frac{n_1}{n_2} \approx \frac{h}{w_1+2w_2+w_3}$, and $n_2$ is
even. Now, we will make new choices for $h$, $w_1$, and $w_3$,
which are close to the old ones. Let $h= n_1$; we distinguish
between two cases of $n_2$:

\begin{enumerate}
\item [(Y.1)] If $n_2 = 4 \omega$ then $w_1 = 2 \omega +1 -w_2$
and $w_3=2 \omega-1-w_2$.

\item [(Y.2)] If $n_2 = 4 \omega +2$ then $w_1 = 2 \omega +3 -w_2$
and $w_3=2 \omega-1-w_2$.
\end{enumerate}

It is easy to verify that the size of the new flipped T,
FT$(h,w_1+w_2+w_3;w_1,w_3)$, $\cS'$, is $n_1 n_2 = p^2-1$,
$\Lambda$ is a lattice tiling for $\cS'$, $(\Lambda,\cS',\delta)$,
$\delta = (0,+1)$, defines a folding, and we can form a doubly
periodic $\cS'$-DDC with it. Hence, we have the following theorem.

\begin{theorem}
Let $n_1$ and $w_2$ be two integers such that $n_2 =
w_1+2w_2+w_3$, $w_1$, $w_3$, are defined by (Y.1), (Y.2), and $n_1
n_2 = p^2 -1$ for some prime number $p$. Then there exists a
doubly periodic $\cS$-DDC, whose shape is a flipped T,
FT$(n_1,w_1+w_2+w_3;w_1,w_3)$, with $p$ dots.
\end{theorem}

\section*{Appendix C}

In this section we demonstrate how Theorem~\ref{thm:infinite} is
applied for several geometric shapes (having the role of $\cQ$ in
the theorem), where our shape $\cS$ in the doubly periodic
$\cS$-DDC is an appropriate corner, flipped T, or quasi-regular
hexagon.

\subsection{Equilateral Triangle}

Let $\cQ$ be an equilateral triangle with sides of length $B$. The
area of $\cQ$ is $\frac{\sqrt{3}}{4} B^2$ and hence an upper bound
on the number of dots in $\cQ$ is $\frac{3^\frac{1}{4}}{2}B + o(B)
= 0.658 B +o(B)$. For our shape $\cS$ we take a flipped T,
FT($\frac{B}{2 \sqrt{2}}, \sqrt{\frac{2}{3}} B ; \frac{B}{2
\sqrt{6}},\frac{B}{2 \sqrt{6}}$) which overlaps in its shorter
base with the base of $\cQ$. These bases of $\cS$ and $\cQ$ share
the same center. The area of $\cS$ is $\frac{\sqrt{3}}{4} B^2$ and
hence the density of the array is $\frac{2}{3^\frac{1}{4}B}$. The
intersection of $\cS$ and $\cQ$, $\Delta (\cQ,\cS)$, equal to
$\frac{3 \sqrt{2} - 2 \sqrt{3}}{2} B^2$. Therefore, a lower bound
on the number of dots in $\cQ$ is $\frac{3 \sqrt{2} - 2
\sqrt{3}}{3^{\frac{1}{4}}} B +o(B) = 0.5916B+o(B)$ and the
resulting packing ratio is 0.899~. The same result can be obtained
by using other structures instead of a flipped T.

\subsection{Isosceles Right Triangle}

Let $\cQ$ be an equilateral triangle with base and height of
length $B$. The area of $\cQ$ is $\frac{1}{2} B^2$ and hence an
upper bound on the number of dots in $\cQ$ is $\frac{1}{\sqrt{2}}B
+ o(B) = 0.707 B +o(B)$. For our shape $\cS$ we take a corner
CR($\sqrt{\frac{2}{3}} B, \sqrt{\frac{2}{3}} B ;
\frac{B}{\sqrt{6}},\frac{B}{\sqrt{6}}$) which overlaps in its two
shorter sides with the base and height of $\cQ$. $\cS$ and $\cQ$
shares the intersection vertex of these sides. The area of $\cS$
is $\frac{1}{2} B^2$ and hence the density of the array is
$\frac{\sqrt{2}}{B}$. The intersection of $\cS$ and $\cQ$, $\Delta
(\cQ,\cS)$, equal to $(\sqrt{6}-2) B^2$. Therefore, a lower bound
on the number of dots in $\cQ$ is $(\sqrt{12} - 2 \sqrt{2})B +o(B)
=0.63567B+o(B)$ and the resulting packing ratio is 0.899~ (exactly
as in the case of an equilateral triangle).

\subsection{Regular Pentagon}

Let $\cQ$ be a pentagon with radius $R$. The area of $\cQ$ is
$\frac{5}{2} \sin \frac{2 \pi}{5}$ and hence an upper bound on the
number of dots in $\cQ$ is $1.54196R +o(R)$. Let $\cS$ be a
quasi-perfect hexagon having a joint base with $\cQ$ and two short
overlapping sides with $\cQ$, where these sides are connected to
this base (see Figure~\ref{fig:pentagon}). The distance between
the base and the diameter of $\cS$ is $aR$, $2 \sin \frac{\pi}{10}
\cos \frac{3 \pi}{10} < a \leq (1 + \sin \frac{3 \pi}{10})/2$. The
length of the base is $2R \sin \frac{\pi}{5}$ and the length of
the diameter of $\cS$ is $2R \sin \frac{\pi}{5} +2 aR \tan
\frac{\pi}{10}$. Hence, the area of $\cS$ is $(4 \sin
\frac{\pi}{5} +2 a \tan \frac{\pi}{10})aR^2$ and the density of
the array is $\frac{1}{\sqrt{4 a \sin \frac{\pi}{5} +2 a^2 \tan
\frac{\pi}{10}}R }$. The area of the intersection between $\cQ$
and $\cS$, $\Delta (\cS,\cQ)$, is computed by subtracting from the
area of $\cS$ the area of the two isosceles triangles $\sigma_1$
and $\sigma_2$. The lower bound on the number of dots is
$\frac{1}{\sqrt{4 a \sin \frac{\pi}{5} +2 a^2 \tan
\frac{\pi}{10}}R } \Delta (\cS,\cQ)$. The maximum on this lower
bound is obtained for $a=0.814853$, i.e., the lower bound on the
number of dots in a pentagon with radius $R$ is $1.45992R+o(R)$
yielding a packing ratio of 0.946795.

\begin{figure}[tb]


\vspace{2cm}

\setlength{\unitlength}{.5mm}
\begin{picture}(40,40)(40,0)
\linethickness{.6 pt}

\put(80,0){\line(1,0){60}} \put(140,0){\line(3.249,10){18.54}}
\put(110,92.3305){\line(10,-7.2654){48.54}}
\put(61.46,57.064){\line(10,7.2654){48.54}}
\put(80,0){\line(-3.249,10){18.54}}

\put(60,9){\makebox(4,25){$\left.\rule{0cm}{24\unitlength}\right\{$}}
\put(52,21){\makebox(0,0){$aR$}} \put(80,0){\line(30,41.2915){30}}
\put(89,21){\makebox(0,0){$R$}}
\put(84,79){\makebox(0,0){$\sigma_1$}}
\put(135,79){\makebox(0,0){$\sigma_2$}}

\linethickness{1.7 pt}
\put(66.4868,41.5893){\line(3.249,10){13.5132}}
\put(80,83.1786){\line(1,0){60}}
\put(140,83.1786){\line(3.249,-10){13.5132}}

\end{picture}

\caption{Quasi-regular hexagon intersecting a regular pentagon}
\label{fig:pentagon}
\end{figure}

\subsection{Regular Heptagon}

Let $\cQ$ be a regular heptagon with radius $R$. The area of $\cQ$
is $\frac{7}{2} \sin \frac{2 \pi}{7} R^2$ and hence an upper bound
on the number of dots in $\cQ$ is $1.65421R + o(R)$. Let $\cS$ be
a quasi-perfect hexagon constructed as follows. We refer to the
sides of $\cQ$ as side 0, side 1, side 2, side 3, side 4, side 5,
side 6, in consecutive order clockwise. Let's denote the six sides
of $\cS$ by side A, side B, side C, side D, side E, side F, in
consecutive order clockwise, where side A is the lower base of
$\cS$. Sides B and C of $\cS$ overlap sides 1 and 2 of $\cQ$,
respectively; sides B and C are longer than sides 1 and 2. The two
bases of $\cS$, sides A and D, have angles $\frac{9 \pi}{14}$ with
sides B and C, respectively. Side A intersect sides 0 and 6 of
$\cQ$; side D intersect sides 3 and 4 of $\cQ$. The length of the
segment, on side 0, from the vertex of the intersection between
sides 0 and 1 and the intersection of side A and side 0 is $xR$.
Finally, sides E and F of $\cS$ are parallel to sides B and C,
respectively; Side E intersect sides 4 and 5; side F intersect
sides 5 and 6. The distance between the vertex of the intersection
between side E and F of $\cS$ and side 5 of $\cQ$ is $aR$.
Computing $|\cS|$, $\Delta (\cS,\cQ)$, and the lower bound on the
number of dots in $\cQ$, i.e., $\frac{\Delta
(\cS,\cQ)}{\sqrt{|\cS|}}$, as functions of $x$ and $a$ implies
that the maximum is obtained for $x=0.432042$ and $a=0.0840633$,
and the lower bound on the number of dots in $\cQ$ is
$1.58844R+o(R)$ yielding a packing ratio of 0.960241.

\subsection{Regular Octagon}

Let $\cQ$ be a regular octagon with radius $R$. The area of $\cQ$
is $4 \sin \frac{\pi}{4} R^2$ and hence an upper bound on the
number of dots in $\cQ$ is $1.68179R + o(R)$. Let $\cS$ be a
quasi-perfect hexagon having a joint diameter of length $2R$ with
$\cQ$ and overlapping four side edges with the $\cQ$. The distance
between the diameter of the hexagon (octagon) and a base of $\cS$
is $\alpha R$. The area of the $\cS$ is $4\alpha R^2 -2\alpha^2
\frac{\sin \frac{\pi}{8}}{\sin \frac{3\pi}{8}} R^2$ and hence the
density of the array is $\frac{1}{\sqrt{4\alpha -2\alpha^2
\frac{\sin \frac{\pi}{8}}{\sin \frac{3\pi}{8}}}R}$. The
intersection between $\cQ$ and $\cS$, $\Delta (\cS,\cQ)$, is $4
\sin \frac{\pi}{4} R^2 - 2(1-\alpha)^2 R^2 \frac{\sin \frac{3
\pi}{8}}{\sin \frac{\pi}{8}}$. Therefore, a lower bound on the
number of dots in the octagon is $\frac{4 \sin \frac{\pi}{4} R -
2(1-\alpha)^2 R \frac{\sin \frac{3 \pi}{8}}{\sin
\frac{\pi}{8}}}{\sqrt{4\alpha -2\alpha^2 \frac{\sin
\frac{\pi}{8}}{\sin \frac{3\pi}{8}}}}$. The maximum is obtained
for $\alpha =0.872852$ and hence a lower bound on the number of
dots is $1.62625R +o(R)$ and the packing ratio is 0.966977 .

\subsection{Regular Nonagon}

Let $\cQ$ be a regular nonagon with radius $R$. Let $\cS$ be a
quasi-regular hexagon with radius $\rho$, where $\rho = \frac{\sin
\frac{ 11 \pi}{18}}{\sin \frac{\pi}{3}} R$, such that $\cQ$ and
$\cS$ share the same center and there is an overlap in three pairs
of edges between $\cQ$ and $\cS$. The area of $\cQ$ is
$\frac{9}{2} \sin \frac{2 \pi}{9} R^2$ and hence an upper bound on
the number of dots in $\cQ$ is $1.700748R + o(R)$. The area of
$\cS$ is $\frac{3 \sqrt{3}}{2} \left( \frac{\sin \frac{ 11
\pi}{18}}{\sin \frac{\pi}{3}} \right)^2 R^2$ and hence the density
of the array is $\frac{\sqrt{2} \sin
\frac{\pi}{3}}{3^{\frac{1}{4}} \sqrt{3} R \sin \frac{ 11
\pi}{18}}$. The area of the intersection between $\cQ$ and $\cS$,
$\Delta (\cS,\cQ)$, is $\frac{3 \sqrt{3}}{2} \left( \frac{\sin
\frac{ 11 \pi}{18}}{\sin \frac{\pi}{3}} \right)^2 R^2 -6
\frac{\sin^2 \frac{\pi}{18} \cos \frac{\pi}{9}}{\sin
\frac{\pi}{3}} R^2=2.8625667R^2$. Therefore, a lower bound on the
number of dots in the nonagon is $1.63672R +o(R)$ and the packing
ratio is 0.96235 .

\subsection{Regular Decagon}


Let $\cQ$ be a regular decagon with radius $R$ with sides 0, 1, 2,
3, 4, 5, 6, 7, 8, 9 in consecutive order clockwise. The area of
$\cQ$ is $5 \sin \frac{\pi}{5} R^2$ and hence an upper bound on
the number of dots in $\cQ$ is $1.71433R + o(R)$. Let $\cS$ a
quasi-perfect hexagon with sides A, B, C, D, E, and F, in
consecutive order clockwise, where A is the lower base of $\cS$.
Sides B and C of $\cS$ overlap with sides 1 and 3 of $\cQ$; sides
E and F of $\cQ$ overlap with sides 6 and 8 of $\cQ$. The two
bases A and D of $\cS$ have distance $aR$ to the diameter of $\cS$
which connects the intersection vertex of sides B and C with the
intersection vertex of sides E and F. The distance between the
diameter and a base (A or D) is $aR$. The area of $\cS$ is $s = 2
(2 \sin \frac{2 \pi}{5} + 2 \frac{\sin \frac{\pi}{10} \sin
\frac{\pi}{5}}{\sin \frac{3 \pi}{10}} - \frac{\sin
\frac{\pi}{5}}{\sin \frac{3 \pi}{10}} a) aR^2$ and the density of
the array is $\frac{1}{\sqrt{s}}$. Finally, $\Delta (\cS,\cQ) =(5
\sin \frac{\pi}{5} - 2 \frac{\sin \frac{2 \pi}{5}}{\sin
\frac{\pi}{10}} (1-a)^2 )R^2$. The lower bound of the number of
dots in $\cQ$ is $\frac{1}{\sqrt{s}} \Delta (\cS,\cQ)$. The
maximum on this lower bound is obtained for $a=0.923286$; the
lower bound is $1.64786R+o(R)$ and the packing ratio is 0.961229 .

\section*{Acknowledgment}
My new interest in DDCs is a consequence of discussions with Simon
Blackburn, Keith Martin, and Maura Paterson during the spring of
2007 on key predistribution for wireless sensor networks. I thank
them all and also for the continuous hospitality which I get
yearly in the Mathematics department in Royal Holloway College.
Without these discussions the current work wouldn't have been
born. Discussions on related problems with Eitan Yaakobi were also
inspirational.


\end{document}